%% file: LinearCoupling_zheng_onecolumn.tex
\documentclass[journal,onecolumn]{IEEEtran}
\usepackage{amssymb,amsfonts,amsmath,amsthm,amscd,dsfont,mathrsfs,mathtools}
\usepackage{graphicx,float,psfrag,pstcol,epsfig,color,subfigure}
\usepackage{psfrag}
\usepackage{pstcol}
\usepackage{import}  
\usepackage{setspace}

\def\BibTeX{{\rm B\kern-.05em{\sc i\kern-.025em b}\kern-.08em
    T\kern-.1667em\lower.7ex\hbox{E}\kern-.125emX}}

\DeclareMathOperator{\diag}{diag}
\newcommand{\cX}{{\cal X}}
\newcommand{\cY}{{\cal Y}}
\newcommand{\cU}{{\cal U}}
\newcommand{\cI}{{\cal I}}

\newcommand{\uX}{\underline{X}}
\newcommand{\ux}{\underline{x}}
\newcommand{\uY}{\underline{Y}}

\newcommand{\uc}{\underline{c}}
\newcommand{\ud}{\underline{d}}
\newcommand{\uu}{\underline{u}}
\newcommand{\uv}{\underline{v}}
\newcommand{\uw}{\underline{w}}
\newcommand{\ue}{\underline{e}}
\newcommand{\ua}{\underline{a}}
\newcommand{\ub}{\underline{b}}
\newcommand{\uth}{\underline{\theta}}

\newcommand{\cS}{{\cal S}}
\newcommand{\cT}{{\cal T}}
\newcommand{\ual}{\underline{\alpha}}

\newcommand{\II}{{\cal S}}
\newcommand{\K}{\psi}
\newcommand{\norm}[1]{\|{#1}\|}
\renewcommand{\vec}[1]{\underline{#1}}

\makeatletter
\newcommand{\vast}{\bBigg@{3}}
\newcommand{\Vast}{\bBigg@{4}}
\makeatother

\theoremstyle{definition}
\newtheorem{definition}{Definition}
\theoremstyle{plain}
\newtheorem{thm}{Theorem}
\theoremstyle{plain}
\newtheorem{prop}{Proposition}
\newtheorem{lemma}{Lemma}
\theoremstyle{plain}
\newtheorem{thmcorol}{Corollary}
\theoremstyle{plain}
\theoremstyle{plain}
\theoremstyle{remark}
\newtheorem{remark}{\bf{Remark}}
\newtheorem{example}{\bf{Example}}

\title{Linear Information Coupling Problems}

\author{Shao-Lun Huang and Lizhong Zheng \\
\thanks{S.-L Huang and L. Zheng are with the Research Laboratory of Electronics at Massachusetts Institute of Technology, Cambridge, USA (Email: $\mathsf{\{ shaolun, lizhong \}@mit.edu}$)}
}

\begin{document}

\maketitle

\begin{abstract}
Many network information theory problems face the similar difficulty of single-letterization. We argue that this is due to the lack of a geometric structure on the space of probability distribution. In this paper, we develop such a structure by assuming that the distributions of interest are close to each other. Under this assumption, the K-L divergence is reduced to the squared Euclidean metric in an Euclidean space. In addition, we construct the notion of coordinate and inner product, which will facilitate solving communication problems. We will present the application of this approach to the point-to-point channel, general broadcast channel, and the multiple access channel (MAC) with the common source. It can be shown that with this approach, information theory problems, such as the single-letterization, can be reduced to some linear algebra problems. Moreover, we show that for the general broadcast channel, transmitting the common message to receivers can be formulated as the trade-off between linear systems. We also provide an example to visualize this trade-off in a geometric way. Finally, for the MAC with the common source, we observe a coherent combining gain due to the cooperation between transmitters, and this gain can be quantified by applying our technique.
\end{abstract}

\section{Introduction}

In this paper, we study a certain class of information theory problems for discrete memoryless communication networks, which we call the \emph{linear information coupling problems}. For a communication network, the corresponding linear information coupling problem asks the question that how we can efficiently transmit a thin layer of information through this network. More rigorously, we assume that there are sequences of input symbols generated at each transmitter from an i.i.d. distribution $P_{X}$. We also assume that the network is composed of some discrete memoryless channels, whose outputs are sequences with an i.i.d. distribution $P_Y$. We take this setup as an operating point. To encode an information $U=u$, we alter some of these input symbols, such that the empirical distribution changes to $P_{X|U=u}$. We insist that for each $u$, $P_{X|U=u}$ is close to $P_{X}$, which means we can only alter a small fraction of the input symbols. Moreover, when averaging over all different values of $u$, the marginal distribution of $X$ remains unchanged. The receivers can then decode the information by distinguishing empirical output distributions with respect to different $u$. The goal of the linear information coupling problem is to design $P_{X|U=u}$ for different $u$, such that the receivers can distinguish different empirical output distributions the most efficiently. Mathematically, for the point-to-point channel with input $X$ and output $Y$, the linear information coupling problem of this channel can be formulated as the multi-letter problem: for a given pair of input and output sequences $X^n, Y^n$, with joint distribution
\begin{align*}
P_{X^n Y^n}(x^n, y^n) = \prod_{i=1}^n P_X(x_i) \cdot P_{Y|X}(y_i|x_i)
\end{align*}
we consider the problem
\begin{align}\label{eq:LICP}
\max_{U \rightarrow X^n \rightarrow Y^n } &\frac{1}{n} I(U;Y^n),\\ \label{eq:LICP_constraint_1}
\mbox{subject to:} \ &\frac{1}{n} I(U;X^n) \leq \delta, \\ \label{eq:LICP_constraint_2}
& \frac{1}{n}  \| P_{X^n|U=u} - P_{X^n} \|^2 = O(\delta), \ \forall u,
\end{align}
where $\delta$ is the amount of information modulated in per input symbol $X$, and is assumed to be small. Here, both $P_{X^n|U=u}$ and $P_{X^n}$ in~\eqref{eq:LICP_constraint_2} are viewed as $| {\cal X} |^n$ dimensional vectors, and the norm square is simply the Euclidean metric.

In fact, the problem \eqref{eq:LICP} is almost the same as the traditional capacity problem
\begin{equation}\label{eq:capacity}
\max_{U \rightarrow X^n \rightarrow Y^n } \frac{1}{n} I(U;Y^n),
\end{equation}
where $U$ is the message transmitted through the channel. This traditional problem has the solution $\max_{P_X} I(X;Y)$ \cite{TJ91}. The difference between \eqref{eq:capacity} and \eqref{eq:LICP} lies in the constraint~\eqref{eq:LICP_constraint_1} and~\eqref{eq:LICP_constraint_2}. Somewhat surprisingly, we will show that with these differences, the linear information coupling problem \eqref{eq:LICP} can be solved quite differently from the corresponding capacity problem \eqref{eq:capacity}. 

The linear information coupling problem \eqref{eq:LICP} indeed captures some fundamental aspects of the traditional capacity problem. We will demonstrate in section \ref{sec:P2P} that the problem \eqref{eq:LICP} is a sub-problem of the capacity problem. In general, the problem \eqref{eq:LICP} is a local version of the global optimization problem \eqref{eq:capacity}, and the solutions of \eqref{eq:LICP} are local optimal solutions of the corresponding capacity problem. In addition, we can ``integral" the solutions of a set of linear information coupling problems back to a solution of the capacity problem.

One important feature of the linear information coupling problems is that when the local assumptions \eqref{eq:LICP_constraint_1} and \eqref{eq:LICP_constraint_2} are added to the problem, there is a systematic approach for single-letterization for general multi-terminal communication problems. We first demonstrate in section \ref{sec:single} that, with a simple linear algebra technique, the linear information coupling problem \eqref{eq:LICP} can be single-letterized to its single-letter version 
\begin{align}\label{eq:MLICP}
\max_{U \rightarrow X \rightarrow Y } & I(U;Y),\\ \notag
\mbox{subject to:} \ & I(U;X) \leq \delta, \\ \notag
& \| P_{X|U=u} - P_{X} \|^2 = O(\delta), \ \forall u,
\end{align}
where $U$ follows the common cardinality bounds. Then, we illustrate in section \ref{sec:broadcast} and \ref{sec:multiple} that for general multi-terminal communication problems, the single-letterization procedure is conceptually the same as the point-to-point channel case. Note that the single-letterization is precisely the difficulty to generalize the conventional capacity results on the point-to-point channels to general multi-terminal problems, this systematic procedure for the linear information coupling problems thus makes these problems particularly attractive.


The main reason that allows this much simpler procedure of single-letterization is that the locality assumptions \eqref{eq:LICP_constraint_1} and \eqref{eq:LICP_constraint_2} fundamentally simplifies the geometric structure of the space of probability distributions. In a nutshell, it allows us to approximate the manifold structure of this space \cite{SH00} by its linear tangent plane. Put it another way, if we view a 1-dimensional family of probability distributions as a parameterized curve in the space of distributions, the locality assumption allows us to focus only on ``straight lines", and further approximate the Fisher information w.r.t. the underlying parameter as a constant along the curve. Such simplification, under different names, has been taken advantage of in several different areas, to produce often the cleanest results, including effcient parameter estimation with large samples, error exponent for very-noisy channels, etc. In the literature of information theory, the work on differential efficiency on investments, by Erkip and Cover \cite{ET98}, which was based on R\'enyi's formulation of maximal correlation \cite{A59-2}, is one of such examples. Some of the connections between these results will be discussed in this paper.

Mathematically, the locality assumption manifests into a quadratic approximation to the Kullback-Leibler (K-L) divergence. When the conditional distributions $P_{X|U=u}$ are close to the empirical distribution $P_{X}$ for all $u$, we can approximate the K-L divergence $D(P_{X|U=u}||P_{X})$, and hence the mutual information $I(U;X)$, by quadratic functions, which turns out to be related to the Euclidean distance between these two distributions. With this local approximation, the space of the input distributions is locally approximated as an Euclidean space around $P_{X}$. Similarly, the space of the output distributions can also be locally approximated as an Euclidean space around $P_{Y}$. We can construct geometric structures in these Euclidean spaces, such as orthonormal bases and inner products. Moreover, it can be shown that the channel behaves as a linear map between the input and output Euclidean spaces. Our purpose is to find the directions to perturb from $P_X$, according to the information $U=u$ to be encoded, in the input distribution space; or equivalently, to design $P_{X|U=u} - P_X$, such that after the channel map, the image of this perturbation at the output distribution space is as large as possible. This turns out to be a linear algebra problem for which even the multi-letter problem can be solved analytically. 

It is worth pointing out that the example on the point-to-point channel, where we linearize the map from the space of distributions on $X$ to that on $Y$, is {\it not} where the power of this local approximation approach lies. In fact, it is well-known that for both the problem without any locality constraint, and that only has \eqref{eq:LICP_constraint_1} but not \eqref{eq:LICP_constraint_2}, can be solved and shown to have single-letter optimal solutions. One can argue that both of these two versions of single-letterization require to establish more involved techniques, and are therefore stronger results than the linear coupling problems. However, these techniques do not generalize easily to multi-terminal problems. In contrast, our solutions to the linear coupling problems can be generalized rather easily. In this paper, we demonstrate this by applying our approach to the general broadcast channels, and show that the localized version of the problem, while does not answer the question of ``capacity region", still can offer insights to the code designs. By doing this, we also point out that the key difficulty of the classical studies on network capacities indeed lies on the non-linear nature of the space of probability distributions, or in other words, the fact that the locality constraint \eqref{eq:LICP_constraint_2} is not used.

The rest of this paper is organized as follows. In section \ref{sec:P2P}, we study the linear information coupling problems for point-to-point channels. We first introduce the notion of local approximation, and show that the K-L divergence can be approximated as the squared Euclidean metric. Then, the single-letter version of the linear information coupling problems will be solved by exploiting the local geometric structure. Moreover, the single-letterization of the linear information coupling problems will be shown to be equivalent to simple linear algebra problems. We will discuss the relation between our work and the capacity results and code designs in section \ref{sec:CALC}, and the relation to the R$\acute{\mbox{e}}$nyi maximal correlation in section \ref{sec:HGR}. Section \ref{sec:broadcast} is dedicated in applying the local approach to the general broadcast channels. It will be shown that the linear information coupling problems of general broadcast channels are different from that for the point-to-point channels in general: the single-letter solutions are not optimal, however finite-letter optimal solutions always exist. 
The application of the local approach to the multiple access channels with common sources is presented in section \ref{sec:multiple}. We show that there are coherent combing gains in transmitting the common sources, and also determine the quantity of these gains. Finally, the conclusion of this paper is given in section \ref{sec:con}.


\section{The Point-to-Point Channel} \label{sec:P2P}

We start with formulating and demonstrating the solutions of the linear information coupling problems for point-to-point channels. For a discrete memoryless point-to-point channel, with input $X \in \cX$ and output $Y \in \cY$, where $\cX$ and $\cY$ are finite sets, let the $|\cY| \times |\cX|$ channel matrix $W$ denote the conditional distributions corresponding to the channel. For this channel, it is known that the capacity is given by
\begin{equation}\label{eq:capacity2}
\max_{P_X} I(X;Y).
\end{equation}
This simple expression is resulted from a multi-letter problem. If we encode a message $U$ in $n$-dimensional vector $X^n$, and decode it from the corresponding $n$-dimensional channel output, we can write the problem as
\begin{equation}\label{eq:multi-letter}
\max_{U \rightarrow X^n \rightarrow Y^n} \frac{1}{n} I(U;Y^n),
\end{equation}
where $U \rightarrow X^n \rightarrow Y^n$ denotes a Markov relation. It turns out that for the point-to-point channel, there is a simple procedure to prove that (\ref{eq:multi-letter}) and (\ref{eq:capacity2}) have the same maximal value \cite{TJ91}:
\begin{align}
\notag
\frac{1}{n} I(U;Y^n) 
&\leq \frac{1}{n} I(X^n;Y^n) \\ \notag
&= \frac{1}{n} \sum_{i} H(Y_i|Y_1^{i-1}) - H(Y_i|X^n,Y^{i-1}) \\ \notag
&\leq \frac{1}{n} \sum_{i} H(Y_i) - H(Y_i|X_i)\\ \label{single-letterization}
&= \frac{1}{n} \sum_{i} I(X_i;Y_i) \leq \max I(X;Y). 
\end{align}
This procedure is known as the \emph{single-letterization}, that is, to reduce a multi-letter optimization problem to a single-letter one. It is a critical step in general capacity problems, since without such a procedure, the optimization problems can potentially be over infinite dimensional spaces, and even numerical solutions of these problems may not be possible. Unfortunately, for general multi-terminal problems, we do not have a systematic way of single-letterization, which is why many of such problems remain open. The most famous examples of such problems are the general (not degraded) broadcast channels. 

In contrast to the capacity problems, we study in this paper an alternative class of problems, called linear information coupling problems. In this section, we consider the linear information coupling problems for point-to-point channels. Assume as before that $X$ and $Y$ are the input and output of a point-to-point channel, the linear information coupling problem of this channel is the following multi-letter optimization problem:
\begin{align}\label{eq:LICP_multi-letter}
\max_{U \rightarrow X^n \rightarrow Y^n } &\frac{1}{n} I(U;Y^n),\\ \label{eq:LICP_multi-letter_constraint_1}
\mbox{subject to:} \ &\frac{1}{n} I(U;X^n) \leq \delta, \\ \label{eq:LICP_multi-letter_constraint_2}
& \frac{1}{n}  \| P_{X^n|U=u} - P_{X^n} \|^2 = O(\delta), \ \forall u,
\end{align}
where $\delta$ is assumed to be small\footnote{In the assumption $\frac{1}{n} I(U;X^n) \leq \delta$, we implicitly assume that $\delta \ll \frac{1}{n}$, for all $n$, so that the approximation in section \ref{sec:local} will be valid for any number of letters.}. The difference between \eqref{eq:multi-letter} and \eqref{eq:LICP_multi-letter} lies in the constraints~\eqref{eq:LICP_multi-letter_constraint_1} and~\eqref{eq:LICP_multi-letter_constraint_2}. In the capacity problem, the entire input sequence is dedicated to encoding $U$; on the other hand, for the linear information coupling problems, we can only alter the input sequence "slightly" to carry the information from $U$. Operationally, we assume that sequences of i.i.d. $P_X$ distributed symbols are transmitted, and the corresponding $P_Y$ distributed symbols are received at the channel output. This can also be viewed as having a pair of jointly distributed multi-source $(X,Y)$ with the distribution $P_{XY}$. Then, we encode the message $U=u$ by altering a small number of symbols in these sequences, such that the empirical distribution changes to $P_{X|U=u}$. As we only alter a small number of symbols, the conditional distribution $P_{X|U=u}$ is close to $P_{X}$. For the rest of this paper, we assume that the marginal distribution $P_{X^n}$ is an i.i.d. distribution over the $n$ letters\footnote{This assumption can be proved to be ``without loss of the optimality'' for some cases \cite{SL08}. In general, it requires a separate optimization, which is not the main issue addressed in this thesis. To that end, we also assume that the given marginal $P_{X^n}$has strictly positive entries.}. 
Our goal is to find the conditional distributions $P_{X|U=u}$ for different values $u$, which satisfy the marginal constraint $P_{X}$, such that a thin layer of information can be conveyed to the $Y$ end the most efficiently.


Although we assume that the operating point has i.i.d. $P_X$ distribution, question remains on whether $P_{X^n | U=u}$ should be i.i.d.. Therefore, \eqref{eq:LICP_multi-letter} has a multi-letter form. In fact, we will show in section \ref{sec:single} that, unlike the capacity problem, the linear information coupling problem allows easy single-letterization, and the optimal conditional distribution $P_{X^n | U=u}$ is indeed i.i.d.. This turns out to be a very important feature of the linear information coupling problems, since the problems are then optimized over finite dimensional spaces.

\subsection{The Local Approximation} \label{sec:local}

The key technique of our approach to solve the linear information coupling problems is to use a local approximation of the K-L divergence. Let $P$ and $Q$ be two distributions over the same alphabet $\cX$, then $D(P \| Q) = \sum_{x} P (x) \log (P (x)/Q(x))$ can be viewed as a measure of distance between these two distributions. However, this distance measure is not symmetric, that is, $D(P \| Q) \neq D(Q \| P)$. 
The situation can be much simplified if $P$ and $Q$ are close. We assume that $Q(x) = P(x) + \epsilon J(x)$, for some small value $\epsilon$, and a function $J : \cX \mapsto \mathbb{R}$. Then, the KL divergence can be written, with the second order Taylor expansion, as
\begin{align}
\notag
D(P \| Q) 
&= -\sum_{x} P(x) \log \frac{Q(x)}{P(x)} \\ \notag
&= -\sum_{x} P(x) \log \left( 1 + \epsilon \cdot \frac{J(x)}{P(x)} \right) \\ 
\label{eqn:epsilon}
&= \frac{1}{2} \epsilon^2 \cdot \sum_{x} \frac{1}{P(x)} J^2(x) + o(\epsilon^2).
\end{align}
We think of $J$ also as a column vector of dimension $| \cX |$, and denote $\sum_{x} J^2(x)/P(x)$ as $\| J \|^2_P$, which is the weighted norm square of the perturbation vector $J$. It is easy to verify here that replacing the weights in this norm by $Q(x)$, or any other distribution in the neighborhood, only results in an $o(\epsilon^2)$ difference. That is, up to the first order approximation, the weights in the norm simply indicate the neighborhood of distributions where the divergence is computed. As a consequence, $D(P \| Q)$ and $D(Q \| P)$ are considered as equal up to the first order approximation.

For convenience of the notations, we define the $\it{weighted \ perturbation \ vector}$ as
\begin{equation}\notag
\K(x) \triangleq \frac{1}{\sqrt{P(x)}} J(x), \ \ \forall x \in \cX,
\end{equation}
or in vector form $\K \triangleq \left[ \sqrt{P}^{-1} \right] J$, where $\left[ \sqrt{P}^{-1} \right]$ represents the diagonal matrix with entries $\left\{ \sqrt{P(x)}^{-1}, \ x \in \cX \right\}$.	This	allows	us	to	write $\| J \|^2_P = \| \K \|^2$, where the last norm is simply the Euclidean norm.

With this definition of the norm on the perturbations of distributions, we can generalize to define the corresponding notion of inner products. Let $Q_i (x) = P (x) + \epsilon \cdot J_i (x)$, $\forall x, i = 1, 2$, we can define
\begin{equation}\notag
\langle J_1 , J_2 \rangle_P \triangleq \sum_{x} \frac{1}{P(x)} J_1 (x) J_2 (x) = \langle \K_1 , \K_2 \rangle,
\end{equation}
where $\K_i = \left[ \sqrt{P}^{-1} \right] J_i$, for $i = 1, 2$. From this, notions of orthogonal perturbations and projections can be similarly defined. The point here is that we can view a neighborhood of distributions as a linear metric space, where each distribution $Q$ is specified by the corresponding weighted perturbation $\K$ from $P$, and define notions of orthonormal basis and coordinates on it.

We now use this new notation to rewrite the linear information coupling problem \eqref{eq:LICP_multi-letter}, which is repeated her convenience. 
\begin{align}\label{eq:LICP_single-letter}
\max_{U \rightarrow X \rightarrow Y } & I(U;Y),\\ \label{eq:LICP_single-letter_constraint_1}
\mbox{subject to:} \ & I(U;X) \leq \delta, \\ \label{eq:LICP_single-letter_constraint_2}
& \| P_{X|U=u} - P_{X} \|^2 = O(\delta), \ \forall u,
\end{align}
For the rest of this paper, we replace the notation $\delta$ in the constraint by $\frac{1}{2}\epsilon^2$, as in the quadratic approximation in \eqref{eqn:epsilon}. We assume that the distribution $P_X$ is given as the operating point. The purpose of \eqref{eq:LICP_single-letter} is to design the distribution $P_U$ and the conditional distributions $P_{X|U=u}$, for different values of $u$, to maximize the mutual information $I(U;Y)$, such that the constraint
\begin{equation}\label{eq:constraint}
I(U;X) = \sum_{u} P_U (u) \cdot D(P_{X|U} (\cdot | u) \| P_X ) \leq \frac{1}{2} \epsilon^2,
\end{equation}
is satisfied, and the marginal distribution $\sum_u P_{U}(u) P_{X|U=u} = P_X$. Since $\epsilon$ is small, from \eqref{eq:constraint} and the local constraint~\eqref{eq:LICP_single-letter_constraint_2}, we can write the conditional distributions $P_{X|U=u}$ as perturbations of $P_X$. Written in vector form, we have $P_{X|U=u} = P_X + \epsilon \cdot J_u$, where $J_u$ is the perturbation vector. With the local approximation on $D(P_{X|U} (\cdot | u) \| P_X )$, the constraint \eqref{eq:constraint} can be written as
\begin{align*}
\frac{1}{2} \epsilon^2 \sum_{u} P_U(u) \cdot \| J_u \|^2_{P_X} + o(\epsilon^2) \leq \frac{1}{2} \epsilon^2,
\end{align*}
which is equivalent to $\sum_{u} P_U(u) \cdot \| J_u \|^2_{P_X} \leq 1$. Moreover, since the conditional distributions $P_{X|U=u}$, for different $u$, have to be valid probability distributions and satisfy the marginal constraint, we have the extra constraints on the perturbation vector $J_u$:
\begin{align} \label{eq:xju}
\sum_x J_u (x) = 0, \forall  u,
\end{align} 
and 
\begin{align} \label{eq:uju}
\sum_u P_U(u) J_u(x) = 0, \forall x.
\end{align}

Next, for each $u$, let $\K_u = \left[ \sqrt{P_X}^{-1} \right] J_u$ be the weighted perturbation vector. Now, we observe that in the output distribution space
\begin{align}
\notag
P_{Y|U=u} 
&= W P_{X|U=u} = W P_X + \epsilon \cdot W J_u \\ \notag
&= P_Y + \epsilon \cdot W \left[ \sqrt{P_X} \right] \K_u,
\end{align}
where the channel applied to an input distribution is simply written as the channel matrix $W$, with dimension $| \cY | \times | \cX |$, multiplying the input distribution as a vector. At this point, we have reduced both the spaces of input and output distributions as linear spaces, and the channel acts as a linear transform between these two spaces. The linear information coupling problem \eqref{eq:LICP_single-letter} can be rewritten as, ignoring the $o(\epsilon^2)$ terms:
\begin{align}\notag
\max. \ &\sum_{u} P_U(u) \cdot \| W J_u \|^2_{P_Y}, \\ \notag
\mbox{subject to:} \ &\sum_{u} P_U(u) \cdot \| J_u \|^2_{P_X} = 1,
\end{align}
or equivalently in terms of Euclidean norms,
\begin{align}\label{eq:bbb200}
\max. \ &\sum_{u} P_U(u) \cdot \left\| \left[\sqrt{P_Y}^{-1} \right] W \left[\sqrt{P_X} \right] \cdot \K_u \right\|^2 \\ \label{eq:bbb}
\mbox{subject to:} \ &\sum_{u} P_U(u) \cdot \| \K_u \|^2 = 1.
\end{align}
The optimization is in the choices of $\K_u$ vectors, which also satisfy the constraints from \eqref{eq:xju}, \eqref{eq:uju}, rewritten   as
\begin{align} \label{eq:bbb_2}
\sum_x \sqrt{P_X(x)} \cdot \K_u(x) = 0, \forall u
\end{align}
and
\begin{align} \label{eq:bbb_3}
\sum_u P_U(u) \sqrt{P_X(x)} \cdot \K_u(x) = 0, \forall x.
\end{align}

The problem \eqref{eq:bbb} is a linear algebra problem. We need to find $P_U$ and a corresponding weighted perturbation vectors $\K_u$ for every $u$, such that the average weighted square norm, as in \eqref{eq:bbb200}, is maximized. For convenience, we write 
\begin{align}
\label{eqn:B}
B \triangleq \left[\sqrt{P_Y}^{-1} \right] W \left[\sqrt{P_X} \right].
\end{align}


Now a simplifying observation is that in both \eqref{eq:bbb200} and \eqref{eq:bbb} the same set of weights $P_U(u)$ are used, thus the problem can be reduced in finding a direction of $\K^*$, which maximizes the ratio $\norm{B \K}/ \norm{\K}$, and the optimal choice of $\K_u$ should be along the direction of this $\K^*$ for every $u$. From the linearity of the problem, scaling $\K_u$ along this direction has no effect on the result. Thus, we can with out loss of optimality pick a simple solution, with $U$ binary equi-probable , and 

\begin{align*}
P_{X|U=0} &= P_X + \epsilon [\sqrt{P_X}] \cdot \K^* \\
P_{X|U=1} &= P_X - \epsilon [\sqrt{P_X}] \cdot \K^*.
\end{align*} 
This makes constraint \eqref{eq:bbb_3} always satisfied. 
 
Figure \ref{fig:Perturb} illustrates this idea from the geometric point of view. We rewrite the optimization problem and the constraints as: 
\begin{align}\label{eq:opt_prob}
\max. \  &\left\| B \cdot \K \right\|^2, \\ \label{eq:opt_prob_2}
\mbox{subject to:} \  &\| \K \|^2 = 1, \\ \label{eq:opt_prob_3}
&\sum_x \sqrt{P_X(x)} \K(x) = 0.
\end{align}
\begin{figure}
\centering 
\subfigure[]{
\includegraphics{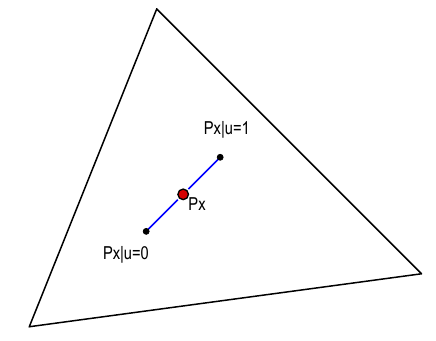} 
\label{fig:Perturb}
}
\subfigure[]{
\includegraphics{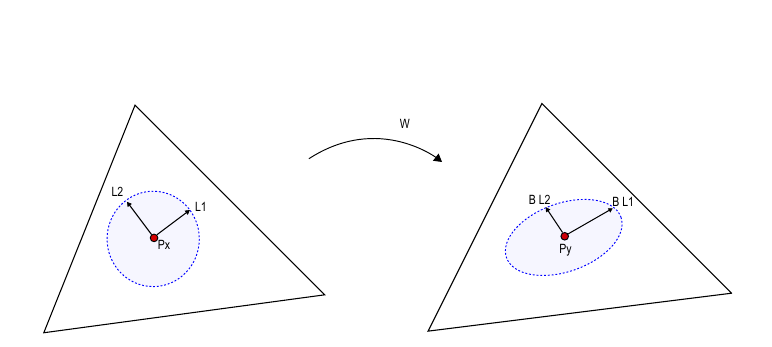}
\label{fig:Channel_Map}
}
\caption{(a) Choice of $P_U$ and $P_{X|U}$ to maintain the marginal $P_X$. (b) Divergence Transition Map as a linear map between two spaces, with right and left singular vectors as orthonormal bases.}
\end{figure}
We call this matrix $B$ as the \emph{divergence transition matrix} (DTM), since it maps divergence in the space of input distributions to that of the output distributions. 

Now, to solve this problem, first note that if we ignore the linear constraint~\eqref{eq:opt_prob_3}, the optimization of~\eqref{eq:opt_prob} is simply choosing $\K$ as the largest right (input) singular vector of $B$ corresponding to the largest singular value. However, this choice might violate~\eqref{eq:opt_prob_3}. We can view \eqref{eq:opt_prob_3} as an orthogonality constraint between $\K$ and a vector $\uv_0 = \left[ \sqrt{P_X}, x \in \cX \right]^T$, and still carry out the optimization. It turns out that the SVD structure of the $B$ matrix makes this particularly simple. 

\begin{lemma}
Let the singular values of the DTM $B$ be $\sigma_0 \geq \sigma_1 \geq \ldots \geq \sigma_m$, with the corresponding right singular vectors $\uv_0, \uv_1, \ldots , \uv_m$, where $m = \min \left\{ |\cX| , |\cY| \right\} -1$, then $\sigma_0 =1$ and $\uv_0 = \left[ \sqrt{P_X}(x), x \in \cX \right]^T$.
\end{lemma}

\begin{proof}
First, it is easy to verify that $B$ has a singular value of $1$, corresponding to left singular vector of $\uw_0= [\sqrt{P_Y(y)}, y \in {\cal Y}]^T$ and right singular vector $\uv_0$, by checking from definition that 
\begin{align*}
B^T\cdot B \cdot \uv_0 = \uv_0, \qquad B\cdot B^T \cdot \uw_0 = \uw_0
\end{align*}

Observe that this $\uv_0$ is an invalid direction to perturb distributions, in that $P_X + \epsilon [\sqrt{P_X}] \cdot \uv_0$ is not a valid distribution. More importantly, any vector orthogonal to $\uv_0$ is a valid perturbation, from \eqref{eq:bbb_2}. That is, any linear combination of the other singular vectors satisfies this constraint. 

To see that all the other singular values must be no larger than $1$, we consider a Markov relation $U \rightarrow X \rightarrow Y$.  Let $U \sim $Bernoulli $(1/2)$, and 
\begin{align*}
P_{X|U=0} = P_X + \epsilon [\sqrt{P_X}] \cdot \K, \quad P_{X|U=1} = P_X - \epsilon [\sqrt{P_X}] \cdot \K
\end{align*}

where $\K$ is orthogonal to $\uv_0$, and hence guarantees the above are valid conditional distributions. Now from the data processing inequality, we have $I(U;Y) \leq I(U;X)$, which implies $\norm{B \K}^2 \leq \norm{\K}^2$. This shows that all other singular values of $B$ must be no larger than $1$. 

\end{proof}






From this lemma, we can conclude that the optimal solution to  \eqref{eq:opt_prob} is to choose $\K$ to be along the right singular vector of $B$ with the second largest singular value, i.e., $\uv_1$.


We can visualize as in Figure \ref{fig:Channel_Map} the orthonormal bases of the input and output spaces, respectively, according to the right and left singular vectors of $B$. The key point here is that while $I(U;X)$ measures how many bits of information is carried in $X$, depending on how the information is modulated, in terms of which direction the corresponding perturbation vector is, the information has different ``visibility" at the receiver end. Picking the weighted perturbation vector to be along $\uv_1$ results in the most ``efficient" way to carry information through the channel.

\begin{remark}
The above arguments imply that 
\begin{equation}\label{eq:sdpi}
I(U;Y) \leq \sigma_1^2 \cdot  I(U;X), 
\end{equation}
where $\sigma_1 \leq 1$ is the second largest singular value of $B$. Thus, comparing to the data processing inequality $I(U;Y) \leq I(U;X)$, \eqref{eq:sdpi} can be viewed as a ``strong data processing" inequality. Moreover, equality can be achieved if and only if for every $u$,  $P_{X|U=u}$ differs from $P_X$ along the weighted direction of $\uv_1$. If the perturbation is along other directions, then the output norm would be reduced even further according to other singular values of $B$. This result gives a clear view of how much information has to be lost when passing through a noisy channel. 

In the literature, there are several other notions of ``strong data processing inequalities". Our result only applies to the case where the locality constraints \eqref{eq:LICP_single-letter_constraint_2} holds. Without this constraint, one can indeed find tighter bounds \cite{Nair13}. The point here is that the local geometric picture is indeed very clean. 


\end{remark}

\begin{remark}

In fact, these ideas are closely related to the method of information geometry \cite{SH00}, which studies the geometric structure of the space of probability distributions. In information geometry, the collection of probability distributions forms a manifold, and the K-L divergence behaves as the distance measure in this manifold. However, the K-L divergence is not symmetric, and this manifold is not flat, but has a rather complicated structure. On the other hand, our approach introduced in this subsection locally approximates this complicated manifold by a tangent hyperplane around $P_X$, which can be viewed as an Euclidean space, and the K-L divergence corresponds to the square norm in this linear space. For the linearized neighborhood around $P_X$, just like any other vector space, one can define many orthonormal bases. Here, we pick the orthonormal basis according to the SVD structure of the DTM B, which is particularly suitable as our goal is to study how much information can be coupled through this channel. This orthonormal basis illustrates the principle directions of conveying information to the receiver end under the channel map, and provides the insights of how to efficiently exploit the channel.

\end{remark}

\begin{remark}
In many network information theory problems, it is required to deal with the tradeoff between multiple K-L divergence (mutual information). Even though K-L divergence is a convex function of both arguments, tradeoff, or linear combinations of multiple convex functions is no longer convex. Therefore, many of such problems are by nature non-convex optimization over potentially high dimensional spaces. This is why analytical solutions can often be hard to find. 

The  local approximation approach is a fundamental simplification of these problems. We approximate the K-L divergence by a quadratic function; and the tradeoff between quadratic functions remains quadratic, which is much easier to deal with. Effectively, our approach tries to find the local optima in such problems, which is a natural step when the problems are non-convex. 
\end{remark}

\begin{remark}
The idea of local analysis on the space of distributions is not new. In fact, it has been used in a wide range of problems, often to produce the cleanest results. One example is non-random parameter estimation \cite{Vantrees}, where asymptotically efficient estimator (achieving the Cramer-Rao bound) always exists when a large number of i.i.d. observations are available. In contrast to the non-asymptotic cases where efficient estimator does not always exists. The underlying reason of this simplicity is that the empirical distribution of a large number of i.i.d. observations lies in a small neighborhood of the true distribution, and local analysis can be employed. 

The contribution of this work is to push this simplification one step further by defining an orthonormal basis on this neighborhood. As we will see in examples, this structure can be quite helpful in the analysis. 
\end{remark}

\begin{example} \label{example1}
\begin{figure}
\centering 
\subfigure[]{
\def\svgwidth{0.35\columnwidth}
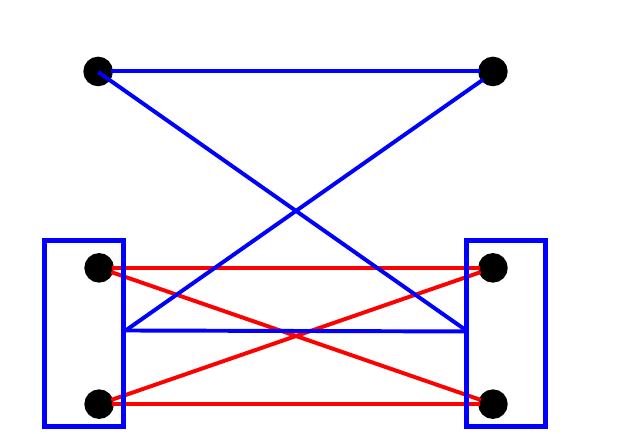
\label{fig:example1_3}
}
\subfigure[]{
\def\svgwidth{0.35\columnwidth}
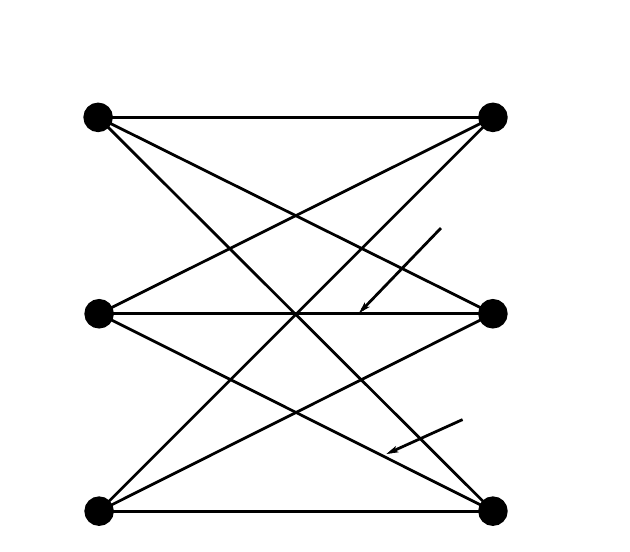
\label{fig:example1}
}
\subfigure[]{
\def\svgwidth{0.35\columnwidth}
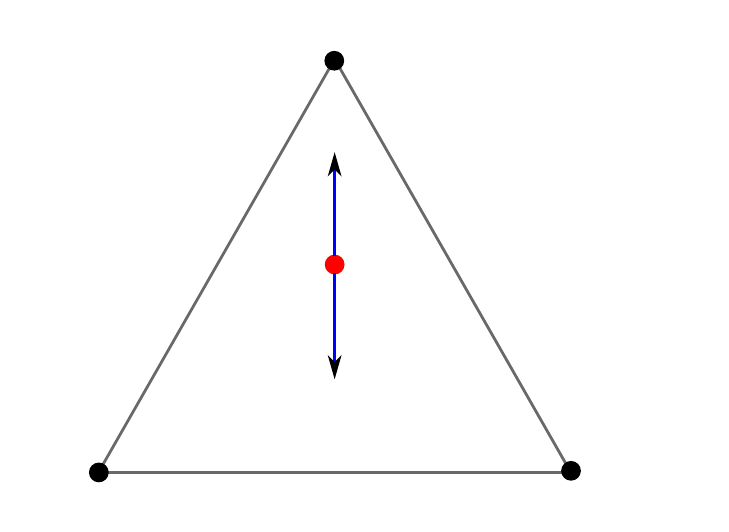
\label{fig:example1_1}
}
\caption{(a) The ternary point-to-point channel that is composed of two binary symmetric channels. (b) The channel transition probability of this ternary channel. (c) The optimal perturbation direction for the ternary channel to convey information to the receiver end. Here, the triangle represents all valid input distributions, and the vertices are the deterministic input distributions of the three input symbols.}
\end{figure}

In this example, we consider a ternary point-to-point channel with input symbols $\cX = \{ 1,2,3 \}$ and output symbols $\cY = \{ 1,2,3 \}$ such that: 
\begin{itemize}
\item [(i)] The sub-channel between the input symbols $\{ 2,3 \}$ and the output symbols $\{ 2,3 \}$ is a binary symmetric channel (BSC) with crossover probability $\frac{1}{2} - \gamma$.
\item [(ii)] If we employ the auxiliary input/output symbol $\bold{0}$ to represent the transmission/receiving of the input/output symbols $2$ and $3$, i.e., $\bold{0} = \{ 2 , 3 \}$, then the sub-channel between the input symbols $\{ \bold{0},1 \}$ and the output symbols $\{ \bold{0},1 \}$ is a BSC with crossover probability $\frac{1}{2} - \eta$.
\end{itemize}
This ternary channel is illustrated in Figure~\ref{fig:example1_3}. Mathematically, the channel transition matrix of this ternary channel can be specified as
\begin{align*}
W =  \left[
      \begin{array}{cccc}
        \frac{1}{2} + \eta & \frac{1}{2} - \eta & \frac{1}{2} - \eta \\
        \frac{1}{4} - \frac{1}{2} \eta & \left( \frac{1}{2} + \eta \right) \left( \frac{1}{2} + \gamma \right) & \left( \frac{1}{2} + \eta \right) \left( \frac{1}{2} - \gamma \right) \\
        \frac{1}{4} - \frac{1}{2} \eta & \left( \frac{1}{2} + \eta \right) \left( \frac{1}{2} - \gamma \right) & \left( \frac{1}{2} + \eta \right) \left( \frac{1}{2} + \gamma \right) \\
      \end{array}
    \right],
\end{align*}
which is illustrated in Figure~\ref{fig:example1}. We do not assume $\gamma$ and $\eta$ to be small, i.e., the channel does not have to be very noisy in this example. 

There are two modes that information can be transmitted through this channel, corresponding to communicating through the two sub-channels described above. That is, to modulate the message in the input symbols $1$ and $\bold{0} = \{ 2,3 \}$ of the BSC $(\frac{1}{2} - \eta)$; or to modulate the message in the input symbols $2$ and $3$ of the BSC $(\frac{1}{2} - \gamma)$. 

Now, to apply our approach, we fix the empirical distribution $P_X$ as $[\frac{1}{2} \ \frac{1}{4} \ \frac{1}{4}]^T$, and the corresponding output distribution $P_Y$ is $[\frac{1}{2} \ \frac{1}{4} \ \frac{1}{4}]^T$. Then, the DTM is
\begin{align*}
B &= \left[
      \begin{array}{cccc}
        \frac{1}{2} + \eta & \frac{1}{2\sqrt{2}} - \frac{1}{\sqrt{2}} \eta & \frac{1}{2\sqrt{2}} - \frac{1}{\sqrt{2}} \eta \\
        \frac{1}{2\sqrt{2}} - \frac{1}{\sqrt{2}} \eta & \left( \frac{1}{2} + \eta \right) \left( \frac{1}{2} + \gamma \right) & \left( \frac{1}{2} + \eta \right) \left( \frac{1}{2} - \gamma \right) \\
        \frac{1}{2\sqrt{2}} - \frac{1}{\sqrt{2}} \eta & \left( \frac{1}{2} + \eta \right) \left( \frac{1}{2} - \gamma \right) & \left( \frac{1}{2} + \eta \right) \left( \frac{1}{2} + \gamma \right) \\
      \end{array}
    \right].
\end{align*}
For this DTM, the singular values are $1$, $2 \eta$, and $\left( 1 + 2 \eta \right) \gamma$, with the corresponding right singular vectors $[ \frac{1}{\sqrt{2}} \ \frac{1}{2} \ \frac{1}{2} ]^T$, $[ \frac{1}{\sqrt{2}} \ \frac{-1}{2} \ \frac{-1}{2} ]^T$, and $[ 0 \ \frac{1}{\sqrt{2}} \ \frac{-1}{\sqrt{2}}]^T$. Translating back to un-weighted perturbations, these corresponds to vectors $[\frac{1}{2}, \frac{1}{4}, \frac{1}{4}]^T, [\frac{1}{2}, -\frac{1}{4}, -\frac{1}{4}]^T$, and $[0, \frac{1}{4}, -\frac{1}{4}]^T$. 

Note here if we perturb $P_X$ along the first vector in any amount, it would result in an invalid distribution. The second and third perturbation vectors, correspond to the two transmission modes described above. For example, if we perturb along the second vector, we would have $P_{X|U = 0} = \left[ \frac{1}{2} + \frac{1}{2} \epsilon \ \frac{1}{4} - \frac{1}{4} \epsilon \ \frac{1}{4} - \frac{1}{4} \epsilon \right]^T$ and $P_{X|U = 1} = \left[ \frac{1}{2} - \frac{1}{2} \epsilon \ \frac{1}{4} + \frac{1}{4} \epsilon \ \frac{1}{4} + \frac{1}{4} \epsilon \right]^T$. This corresponds to increasing or decreasing the fraction of symbol $1$ transmitted, according to the value of $U$. 

The efficiency of these two modes depends on the two corresponding singular values. Here, the comparison is more in favor of the first mode, since given $P_X=[\frac{1}{2}, \frac{1}{4}, \frac{1}{4}]^T$, we can transmit symbol $2$ or $3$ only half of the time. 

The point of this example is that for general problems, where we cannot identify naturally separable transmission modes by inspection, the SVD structure of the DTM matrix can always help us to do that like in this simple example.

\end{example}



\subsection{The Single-Letterization} \label{sec:single}

The most important feature of the linear information coupling problem \eqref{eq:LICP_multi-letter} is that the single-letterization is simple. To illustrate the idea, we first consider a $2$-letter version of the point-to-point channel:
\begin{align}\label{eq:LICP_single-leterization}
\max_{U \rightarrow X^2 \rightarrow Y^2 } &\frac{1}{2} I(U;Y^2),\\ \notag
\mbox{subject to:} \ &\frac{1}{2} I(U;X^2) \leq \frac{1}{2} \epsilon^2, \\ \notag
& \frac{1}{2}  \| P_{X^2|U=u} - P_{X^2} \|^2 = O(\epsilon^2), \ \forall u,
\end{align}
Let $P_X$, $P_Y$, $W$, and $B$ be the input and output distributions, channel matrix, and the DTM, respectively, for the single letter version of the problem. Then, the 2-letter problem has $P^{(2)}_X = P_X \otimes P_X$, $P^{(2)}_Y = P_Y \otimes P_Y$, and $W^{(2)} = W \otimes W$, where $\otimes$ denotes the Kronecker product. As a result, the new DTM is $B^{(2)} = B \otimes B$. Thus, the optimization in \eqref{eq:LICP_single-leterization} has exactly the same form as in \eqref{eq:LICP_single-letter}, where the only difference is that we need to find the SVD of $B^{(2)}$ instead of $B$. For that, we have the following lemma, the proof of which is omitted. 

\begin{lemma}\label{lem:1}
Let $\uv_i$ and $\uv_j$ denote two right (or left) singular vectors of $B$ with singular values $\sigma_i$ and $\sigma_j$. Then, $\uv_i \otimes \uv_j$ is a right (or left) singular vector of $B^{(2)}$ and the corresponding singular value is $\sigma_i \cdot \sigma_j$.
\end{lemma}

Recall that the largest singular value of $B$ is $\mu_0 = 1$, with the right singular vector $\uv_0 = \left[ \sqrt{P_X}, x \in \cX \right]^T$, which corresponds to the direction orthogonal to the distribution simplex. This implies that the largest singular value of $B^{(2)}$ is also 1, corresponding to the singular vector $\uv_0 \otimes \uv_0$, which is again orthogonal to all valid choices of the weighted perturbation vectors.

The second largest singular value of $B^{(2)}$ is a tie between $\sigma_0 \cdot \sigma_1$ and $\sigma_1 \cdot \sigma_0$, with right singular vectors $\uv_0 \otimes \uv_1$ and $\uv_1 \otimes \uv_0$, where $\sigma_1$ is the second largest singular value of $B$, and $\uv_1$ is the corresponding right singular vector. The optimal solution of (\ref{eq:LICP_single-leterization}) is thus the weighted perturbation vectors to be along the subspace spanned by these two vectors. This can be written as
\begin{align}
\label{eq:aaa}
P_{\uX|U=u} 
&= P_X \otimes P_X + \left[\sqrt{P_X \otimes P_X} \right] \cdot \left( \epsilon \uv_0 \otimes \uv_1 + \epsilon' \uv_1 \otimes \uv_0 \right) \\ \label{eq:aa}
&= \left( P_X + \epsilon' \left[\sqrt{P_X}\right] \uv_1 \right) \otimes \left( P_X + \epsilon \left[\sqrt{P_X}\right] \uv_1 \right) + O(\epsilon^2),
\end{align}
where (\ref{eq:aa}) comes from noting that the vector $\uv_0 = \left[\sqrt{P_X}, x \in \cX \right]^T$, and adding the appropriate cross term for factorization. Here, we assume that $\epsilon$ and $\epsilon'$ are of the same order, which makes the cross term $O(\epsilon^2)$. This means that up to the first order approximation, the optimal choice of $P_{X^2|U=u}$, for any value of $u$, has a product form, i.e., the two transmitted symbols in $X^2$ are conditionally independent given $U$. With a simple time-sharing argument, we can see that it is optima to set $\epsilon = \epsilon'$. This implies that picking $P_{X^2|U=u}$ to be i.i.d. over the two symbols achieves the optimal, with the approximation in \eqref{eq:aa}. 

Finally, by considering the $n^{th}$ Kronecker product, we can generalize this procedure to the single-letterization of the $n$-letter problem \eqref{eq:LICP_multi-letter}. 


\begin{remark}
This proof of showing the single-letter optimality is simple. All we have used is the fact that the singular vectors of $B^{(2)}$ corresponding to the second largest singular value has a special form, $\uv_0 \otimes \uv_1$ or $\uv_1 \otimes \uv_0$. We can visualize this as follows. The space of $2$-letter joint distributions $P_{X^2|U=u}$ has $\left( |\cX|^2 - 1 \right)$ dimensions. Around the i.i.d. marginal distribution $P_X \otimes P_X$, there is a $2 \cdot \left( |\cX| -1 \right)$-dimensional subspace, such that the distributions in this subspace take the product form $Q_1 \otimes Q_2$, for some distributions $Q_1$ and $Q_2$ around $P_X$. These distributions can be written as perturbations from $P_X \otimes P_X$, with the weighted perturbations of the form $\uv_0 \otimes \uv + \uv' \otimes \uv_0$, for some $\uv$ and $\uv'$ orthogonal to $\uv_0$. The above argument simply verifies that the optimal solution to~\eqref{eq:LICP_single-leterization}, which is the singular vectors of the $B^{(2)}$ matrix, has this form. We argue that this geometric view was not clear from the classical proofs of single-letterization. Moreover, it turns out that this procedure can be applied to more general problems. In section \ref{sec:broadcast} and \ref{sec:multiple}, we will demonstrate that in quite a few other multi-terminal problems, the similar structure can be proved and used for single-letterization. 


We would like to emphasize that the advantage of our approach is that it does not require any constructive proving technique, such as constructing auxiliary random variables. For any given problem, one can follow essentially the same procedure to find out the SVD structure of the corresponding DTM. The result either gives a proof of the local optimality of the single-letter solutions or disproves it without any ambiguity.
\end{remark}

\subsection{Remarks On The Local Constraint~\eqref{eq:LICP_multi-letter_constraint_2}} \label{sec:remarks}

Note that in our linear information coupling problem~\eqref{eq:LICP}, we not only assume that the mutual information $\frac{1}{n}I(U;X^n)$ is small, but also restrict that the conditional distributions $P_{X^n|U=u}$ satisfy the local constraint $\frac{1}{n}  \| P_{X^n|U=u} - P_{X^n} \|^2 = O(\epsilon^2)$, for all $u$. With the local constraint~\eqref{eq:LICP_multi-letter_constraint_2} on $P_{X^n|U=u}$, we can then guarantee the validity of the local approximation of K-L divergence in section~\ref{sec:local}. Therefore, the local constraint~\eqref{eq:LICP_multi-letter_constraint_2} is indeed critical in order to obtain the linearized geometric structure.


Importantly, the local constraint~\eqref{eq:LICP_multi-letter_constraint_2} has to be specified independently from the constraint~\eqref{eq:LICP_multi-letter_constraint_1}, because assuming $\frac{1}{n}I(U;X^n)$ to be small does not necessarily imply the local constraint on all the conditional distributions. It is possible that the joint distribution $P_{X^nU}$ satisfies $\frac{1}{n}I(U;X^n) \leq \frac{1}{2} \epsilon^2$, but the conditional distributions $P_{X^n|U=u}$ are far from $P_{X^n}$ for some $u$'s, with the corresponding $P_U(u)$'s very small. Therefore, optimizing the mutual information $\frac{1}{n}I(U;Y^n)$ with only the constraint $\frac{1}{n}I(U;X^n) \leq \frac{1}{2} \epsilon^2$ can be a different problem from our linear information coupling problem.

In fact, Ahlswede and G$\acute{\mbox{a}}$cs in~\cite{Ahlswede76}, and a recent paper by Nair \emph{et al.}~\cite{Nair13} considered the following quantity
\begin{align} \label{eq:s(x,y)}
s(X^n,Y^n) = \lim_{I(U;X^n) \rightarrow 0} \sup_{U \rightarrow X^n \rightarrow Y^n} \frac{I(U;Y^n)}{I(U;X^n)},
\end{align}
where they established two important statements:
\begin{itemize}
\item [(i)] For i.i.d. $P_{X^nY^n} = P^n_{XY}$, the $s(X^n,Y^n)$ can be tensorized (single-letterized), i.e., $s(X^n,Y^n) = s(X,Y)$.
\item [(ii)] In general, $s(X^n,Y^n)$ can be strictly larger than the $\sigma_1^2$, where $\sigma_1$ is the second largest singular value of the divergence transition matrix $B$ that we developed in section~\ref{sec:local}.
\end{itemize}

The statement (i) is an important property of $s(X^n,Y^n)$, because it addresses the single-letterization of the multi-letter problem in information theory, which reduces a computationally impossible problem to a computable one. On the other hand, the $\sigma_1^2$ we considered in our local geometry can also be tensorized by a linear algebra approach as we showed in section~\ref{sec:single}. So, both $s(X^n,Y^n)$ and $\sigma_1^2$ have the tensorization property in this case of point-to-point communications.

Moreover, the statement (ii) implies that, without the local constraint, the optimal achievable information rate $\frac{1}{n}I(U;Y^n)$, subject to $\frac{1}{n}I(U;X^n) \leq \frac{1}{2} \epsilon^2$, is $s(X^n,Y^n) \cdot \frac{1}{2} \epsilon^2$. This is strictly better than the case with the local constraint, where the optimal achievable information rate is $\sigma_1^2 \cdot \frac{1}{2} \epsilon^2$. In that sense, $s(X^n,Y^n)$ is a strictly more meaningful quantity for this problem. 


However, we would like to point out that it is still worth considering the quantity $\sigma_1^2$. The value of the development we have shown so far in this paper does not lie in the tensorization result, but rather in the geometric method we used to arrive at this result. As we have stated in several different ways, the local assumptions fundamentally simplifies the problems involving tradeoff between multiple mutual information, which is the core of many problems seeking to find the multi-terminal capacity-regions. By taking this simplification, we focus on finding the local optimal solutions to the problem. In some sense, we have thus given up the hope of finding the globally optimal solution, and hence the hope of finding in general the capacity regions in the classical formulations. In return, the geometric insights from this approach does offer valuable guidance to code designs; and more importantly, it turns out that this simplification allows us to generalize our technique to the studies of network problems, in a conceptually straight forward way, which will be demonstrated in Section \ref{sec:broadcast}. In contrast, the technique used in the non-local version of the problem, such as that used in proving the tensorization of $s(X^n, Y^n)$, is intrinsically based on the idea used in the study of degraded broadcast channels, and is difficult to generalize beyond a handful of canonical examples.

Before moving to the more interesting multi-terminal problems, we discuss in the rest two subsections that how the linear information coupling problems can be connected to the capacity problems, and also the relation between the linear information coupling problems and the R$\acute{\mbox{e}}$nyi maximal correlation. Readers, who are only interested in the application of our local approach to the multi-terminal problems, can directly turn to the section \ref{sec:broadcast} and \ref{sec:multiple}.


\begin{figure*}
\centering 
\def\svgwidth{0.67\columnwidth}
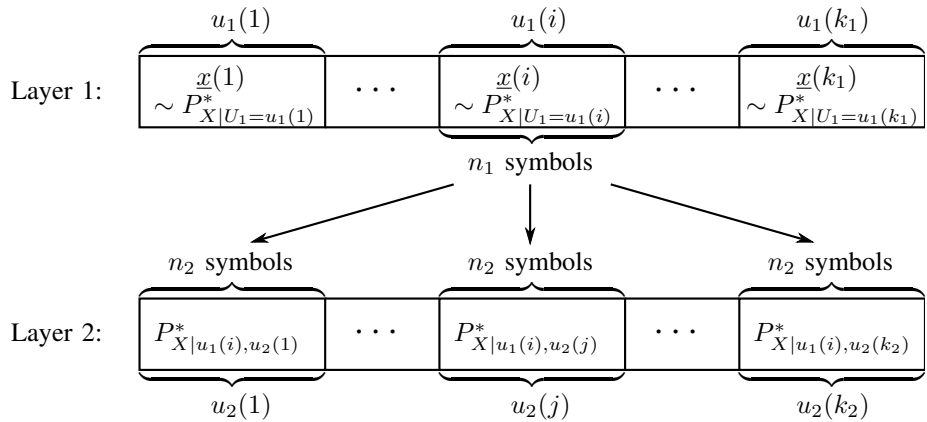
\caption{The empirical distribution of different sub-blocks in each layer after encoding.}
\label{fig:layer_coding}
\end{figure*}

\subsection{Capacity Achieving Layered Codes} \label{sec:CALC}

In this subsection, we discuss one operational meaning of the linear information coupling problem, and try to connect to that of the conventional capacity problem for the point-to-point channels. To do that, we construct a channel code as superposition of many layers of codes, each layer constructed from the solution of of a specific linear coupling problem. This construction is hardly useful in any practical situation, but rather serves as a conceptual tool to connect the two problems.

Let us start from the one-layer problem of this coding scheme. For a point-to-point channel with a transmitter $X$ and a receiver $Y$, the goal of the one-layer problem is to efficiently transmit information through the Markov relation $U_1 \rightarrow X \rightarrow Y$, subject to the constraint $I(U_1 ; X) \leq \frac{1}{2} \epsilon^2$, and the local constraint $\| P_{X|U_1=u_1} - P_X \|^2 = O(\epsilon^2) $. From the analyses of the linear information coupling problem, we know how to find the optimal $P^*_{X|U_1 = u_1}$ and $P^*_{U_1}$ to achieve the solution
\begin{align}\label{eq:R*}
r^*_1 = \mathop{ \max_{U_1 \rightarrow X \rightarrow Y : \  I(U_1;X) \leq \frac{1}{2} \epsilon^2,} }_ {\| P_{X|U_1=u_1} - P_X \|^2 = O(\epsilon^2), \ \forall u_1} I(U_1;Y) 
\end{align}
Now, we propose the following coding scheme to explain the operational meaning of this solution. Suppose that there is a block\footnote{In this paper, all the ``block length" and ``number of sub-blocks" are assumed to be large.} of $n_1 \cdot k_1$ i.i.d. $P_X$ distributed input symbols $ \ux(1), \ldots , \ux(k_1)$ generated at the transmitter, where $\ux(i) \in \cX^{n_1}$ represents a sub-block of $n_1$ input symbols $x_1(i) , \ldots , x_{n_1}(i) $, for $1 \leq i \leq k_1$. Then, we ``encode" a binary codeword $u_1(1),  \ldots , u_1(k_1)$, with empirical distribution $P^*_{U_1}$, into this input symbol block by altering some of the symbols, such that the empirical distribution of each sub-block $\ux(i)$ changes to $P^*_{X|U_1 = u_1(i)}$. Note that the empirical distribution of the entire symbol block remains approximately the same as $P^*_{X}$. The receiver decodes this codeword according to different empirical output distributions of the $k_1$ sub-blocks. From \eqref{eq:R*}, there exists binary block codes $u_1(1),  \ldots , u_1(k_1)$ with rate $R^*_1 = n_1 \cdot r^*_1$ bits/$U_1$ symbol, which can be reliably transmitted and decoded by using the above coding scheme. The empirical distributions of different blocks of input symbols, after this encoding procedure, are illustrated in Figure~\ref{fig:layer_coding}.

Now, we can add another layer to the one-layer problem. Theoretically, this is to consider a new set of linear information coupling problems
\begin{align} \label{eq:R*2}
r_2^*(u_1) =   \mathop{ \max_{(U_1, U_2) \rightarrow X \rightarrow Y :  \  I(U_2;X|U_1 = u_1) \leq \frac{1}{2} \epsilon^2,} }_{ \| P_{X|U_1=u_1,U_2=u_2} - P^*_{X|U_1=u_1} \|^2 = O(\epsilon^2), \ \forall u_2 } I(U_2;Y|U_1 = u_1), 
\end{align}
where the conditional distribution of $X$ given $U_1 = u_1$ is specified as $P^*_{X|U_1 = u_1}$. We can solve \eqref{eq:R*2} with the same procedure as \eqref{eq:R*}, and find the optimal solutions $P^*_{X|U_1 = u_1 , U_2 = u_2}$ and $P^*_{U_2|U_1 = u_1}$. 

Then, we can encode this one more layer of codewords to the original layer with a similar coding scheme. To do this, we further divide each sub-block $\ux(i)$ into $k_2$ small sub-blocks, and each of the small sub-block has $n_2$ symbols, where $n_2 \cdot k_2 = n_1$. Then, for a binary code $u_2(1),  \ldots , u_2(k_2)$ with rate $R_2^*(u_1(i)) = n_2 \cdot r_2^*(u_1(i))$ bits/$U_2$ symbol, where the distribution of the bits in the codewords is $P^*_{U_2|U_1 = u_1(i)}$, we encode the codewords into small sub-blocks of $\ux(i)$ by exactly the same coding scheme as the one-layer problem. The transmission rate of this coding scheme over the entire input symbol block $ \ux(1), \ldots , \ux(k_1)$ is then 
\begin{align*}
\sum_{u_1}  r_2^*(u_1) P^*_{U_1} (u_1) = I(U_2;Y|U_1) \quad \mbox{bits/transmission.}
\end{align*}
After this, the empirical distribution of the $j$-th small sub-block of $\ux(i)$ changes to $P^*_{X|U_1 = u_1(i), U_2 = u_2(j)}$, which is illustrated in Figure~\ref{fig:layer_coding}. On the other hand, the empirical distribution of the entire $\ux(i)$ remains approximately the same as $P^*_{X|U_1 = u_1(i)}$. Thus, the decoding of the codewords $u_1(1),  \ldots , u_1(k_1)$ of the first layer is not effected by adding the second layer, and can be proceeded as in the one-layer problem. The codewords of the second layer are then decoded after the first layer is decoded.

We can keep adding layers by recursively solving new linear information coupling problems, and sequentially applying the above layered coding scheme. Assuming that there is a sequence of messages $U_1,U_2,...,U_K$ that we want to encode. 
First, we can find a perturbation of the $P_X$ distribution according to the value of $U_1$ by solving the corresponding linear information coupling problem. Then, by solving the new set of information coupling problems conditioned on each value of $U_1 = u_1$, we can find further perturbations of that according to the value of $U_2$, and so on. The corresponding perturbations in the output distribution space is illustrated in Figure~\ref{fig:Layer}. 

\begin{figure}
\centering 
\includegraphics{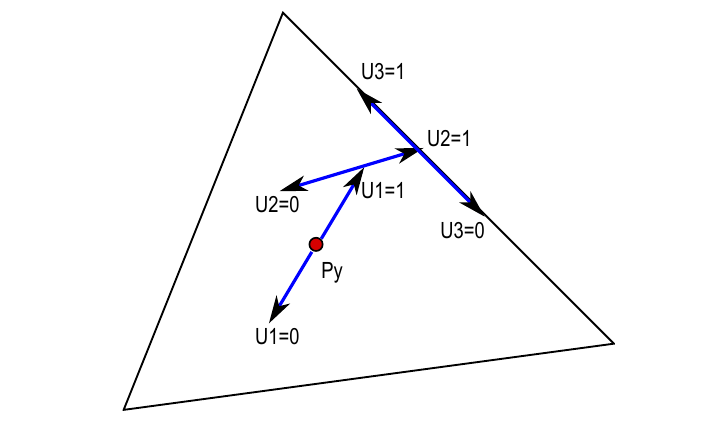}
\caption{A layered approach for the coding problem.}
\label{fig:Layer}
\end{figure}

Note that for each layer, say, layer $k$, while the channel matrix $W$ remains the same, as we perturb only a sub-block of the symbols, the operating point is the empirical distribution of the sub-block $P_{X|U_1^{k-1}=u_1^{k-1}}$, which differs from the original $P_X$. Thus the resulting $B$ matrix is also different. By this construction, we demonstrate that a channel code can be constructed through a sequence of layers, and thus can be viewed as an ``integral" of a sequence of local perturbation problems. There are however several important differences between the two problems. 

First, the most obvious issue is that since we divide the sub-blocks further with each layer of information, it appears that the sub-blocks gets very short as the number of layers increases. In fact, we can group all the sub-blocks with the same empirical distribution together before further division. For example, all the sub-blocks in Figure \ref{fig:Layer} with length $n_1$ with the corresponding $u_1(i)=0$ can be grouped together before further dividing. In more general cases, some sub-blocks from different branches might end up having the same or approximately the same empirical distributions, and thus grouped together. Thus, the total number of sub-blocks at each layer is limited by the granularity of empirical distributions we choose to group sub-blocks, and does not have to increase exponentially with the number of layers.


The second issue, as shown in Figure \ref{fig:Layer}, the valid choices of distributions on the channel output must be in a convex region. For a given channel matrix $W$, whose column vectors are the conditional distributions of the output $Y$, conditioned on different values of the channel input $X = x$, the output distributions must belong to the convex region specified by these column vectors. As we add more layers, at some point the boundary of this convex region is reached. From which point, further layering is restricted to be along the hypersurface of this convex region. Conceptually, there is not much difference, since moving the output distributions on the hypersurface corresponds to not use a subset of the input alphabet. Hence, a local problem can in principle be written out with a reduced input alphabet. This can indeed be done in some special cases \cite{SL08}. However, for general problems, especially multi-letter problems, specifying this high dimensional convex region and all its boundary constraints seems to be a reason that forbids general analytical solutions.

Finally, the most significant difference between the two problems is that although a channel code is constructed as superposition of many layers of codes, optimizing the coupling efficiency at each individual layer, i.e., using the solution we specified in \eqref{eq:sdpi}, for  each layer, does not necessarily yield the optimality of the overall code. This is because at each layer we not only convey the corresponding layer of information through the channel, but also the resulting empirical distributions of sub-blocks become the operating points for the future layers. Thus, the construction of overall channel code can be viewed as a dynamic programming, where each layer not only needs to carry as much information as possible, but also needs to set up favorable operating points for the future layers. Our solution based on SVD analysis can thus be viewed as a greedy solution to the channel coding problem. 

For some special cases, especially if the channel considered is very noisy, we can indeed use the above approach to design a capacity achieving channel code. The following is one of such examples.


\begin{example} 
We continue to use the example \ref{example1}, but make the channel very noisy by setting both the parameters $\gamma$ and $\eta$ to be close to $0$. We assume $\gamma/\eta$ remains constant, and consider only the case that $\gamma < \eta$. We apply the layered coding scheme to construct a capacity achieving code. 

First, ignoring the higher order terms, the channel capacity of the ternary channel in Figure~\ref{fig:example1} is $2 \eta^2 + \left( \frac{1}{2} + \eta \right) \gamma^2$, with the optimal input distribution $P_X = [\frac{1}{2} \ \frac{1}{4} \ \frac{1}{4}]^T$.  From example~\ref{example1}, we know that when $\gamma > \eta$, 
the optimal perturbation vector is $[ \frac{1}{2} \ \frac{-1}{4} \ \frac{-1}{4} ]^T$, and the corresponding conditional distributions are $P_{X|U_1 = 0} = \left[ \frac{1}{2} + \frac{1}{2} \epsilon \ \frac{1}{4} - \frac{1}{4} \epsilon \ \frac{1}{4} - \frac{1}{4} \epsilon \right]^T$ and $P_{X|U_1 = 1} = \left[ \frac{1}{2} - \frac{1}{2} \epsilon \ \frac{1}{4} + \frac{1}{4} \epsilon \ \frac{1}{4} + \frac{1}{4} \epsilon \right]^T$. To apply the layered coding scheme, we keep increasing the perturbation vector until the boundary is reached, i.e., increasing\footnote{Since we assume both $\eta$ and $\gamma$ are small, the local approximation of all divergence and mutual information of interests remains valid even if $\epsilon$ is not small. This is why we can increase $\epsilon$ here from a small number to $1$ without violating the local approximation.} $\epsilon$ to $1$. Then, the conditional distribution $P_{X|U_1 = 0}$ reaches the vertex $[1 \ 0 \ 0]^T$, and $P_{X|U_1 = 1}$ reaches the boundary at $[0 \ \frac{1}{2} \ \frac{1}{2}]^T$. This is shown in Figure~\ref{fig:example1_2}. The achievable information rate by the first layer of perturbation is $I(U_1;Y) = \frac{1}{2} \epsilon^2 (2\eta)^2 = 2 \eta^2$. 

To achieve this rate, we divide the $n$ bit codeword into $k_1$ sub-blocks, each with $n_1$ bits, and use a binary code of length $k_1$, with rate $R_1= (2\eta^2) \cdot n_1$. We choose $n_1$ appropriately to make sure $R_1<1$. The total number of information bits encoded is $k_1R_1 =n\cdot  2\eta^2$. The $k_1$ coded bits are assigned to each sub-block. The sub-blocks corresponding to a coded bit of $0$ are filled with channel symbol '1's. The rest of sub-blocks should be filled with half '2's and half '3's. We group these sub-blocks together, of total length $n/2$, for the second layer of information. To do that, we further divide these $n/2$ symbols into $k_2$ sub-blocks, each with $n_2$ symbols.

We perturb the conditional distribution $P_{X|U_1 = 1} = [0 \ \frac{1}{2} \ \frac{1}{2}]^T$. Note that this distribution has already reached the boundary of the simplex, and we cannot further reduce the probability of '0'. Thus, the perturbation is along this boundary. This corresponds to a linear information coupling problem with reduced input alphabet of just $\{2,3\}$.  Therefore, the DTM of this problem has reduced dimension, and can be explicitly computed as
\begin{align*}
B_2 &= \left[
      \begin{array}{ccc}
        \sqrt{\frac{1}{4} - \frac{1}{2} \eta} & \sqrt{\frac{1}{4} - \frac{1}{2} \eta} \\
        \sqrt{\frac{1}{2} + \eta} \left( \frac{1}{2} + \gamma \right) & \sqrt{\frac{1}{2} + \eta} \left( \frac{1}{2} - \gamma \right) \\
        \sqrt{\frac{1}{2} + \eta} \left( \frac{1}{2} - \gamma \right) & \sqrt{\frac{1}{2} + \eta} \left( \frac{1}{2} + \gamma \right) \\
      \end{array}
    \right].
\end{align*}
The second largest singular value of this DTM is $\sigma_1 = \sqrt{2 + 4\eta} \cdot \gamma$, and the corresponding singular vector is $\uv_1 = [ 0 \ \frac{1}{\sqrt{2}} \ \frac{-1}{\sqrt{2}} ]^T$. Observe that this new singular value is smaller than that for the first layer, indicating less efficient coupling of information. The optimal perturbation vector is $[ 0 \ \frac{1}{2} \ \frac{-1}{2} ]^T$, and the conditional distributions are $P_{X|U_1 = 1, U_2 = 0} = \left[ 0 \ \frac{1}{2} + \frac{1}{2} \epsilon \ \frac{1}{2} - \frac{1}{2} \epsilon \right]^T$ and $P_{X|U_1 = 1, U_2 = 1} = \left[ 0 \ \frac{1}{2} - \frac{1}{2} \epsilon \ \frac{1}{2} + \frac{1}{2} \epsilon \right]^T$. We choose $\epsilon=1$, so the perturbed distributions reach the two vertices $[0 \ 1 \ 0]^T$ and $[0 \ 0 \ 1]^T$, as shown in Figure~\ref{fig:example1_2}.

The achievable information rate by the second layer of perturbation is 
\begin{align*}
I(U_2;Y|U_1) = I(U_2;Y|U_1 = 1) \cdot P(U_1 = 1) = \frac{1}{2} \epsilon^2 (\sqrt{2 + 4\eta} \cdot \gamma)^2 \cdot \frac{1}{2}  = \left( \frac{1}{2} + \eta \right) \gamma^2.
\end{align*}

In terms of the code construction, this second layer of information is conveyed by using a binary code of length $k_2$, with rate $R_2 = 2 \cdot n_2 \cdot \left(\frac{1}{2} + \eta \right) \gamma^2$. The total number of information bits carried is thus $k_2 \cdot R_2 = n \cdot \left(\frac{1}{2} + \eta \right) \gamma^2$.  These $k_2$ coded bits are assigned to the corresponding sub-blocks. Those assigned with a $0$ are filled with transmitted symbol '2's, and the others with '3's. 

After these two layers of perturbations, all the conditional distributions reach the vertices, and the total achievable information rate is $2 \eta^2 + \left( \frac{1}{2} + \eta \right) \gamma^2$, which achieves the channel capacity of this ternary channel. 

Note that resulting code repeats symbol '1' for $n_1$ times, symbols '2' and '3' for $n_2$ times. Such repetition is indeed expected for codes that achieve the capacity for very noisy channels. 


\begin{figure}
\centering
\def\svgwidth{0.35\columnwidth}
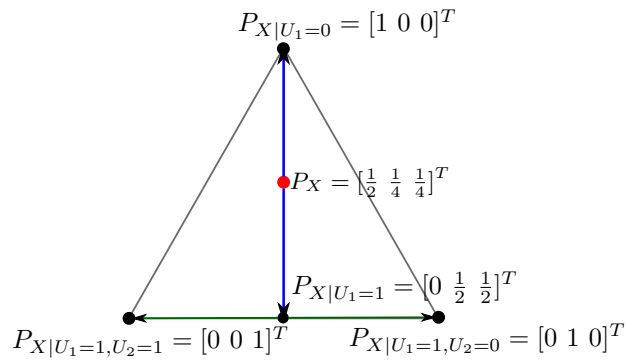
\caption{In the first layer, the input distribution $P_X$ is perturbed to $P_{X|U_1 = 0}$ and $P_{X|U_1 = 1}$, where $P_{X|U_1 = 0}$ reaches one of the vertices, and $P_{X|U_1 = 1}$ reaches the boundary. In the second layer, the distribution $P_{X|U_1 = 1}$ is further perturbed to $P_{X|U_1 = 1, U_2 = 0}$ and $P_{X|U_1 = 1, U_2 = 1}$, where the rest two vertices are reached.}
\label{fig:example1_2}
\end{figure}

\end{example}


\subsection{The relation to R$\acute{\mbox{e}}$nyi maximal correlation} \label{sec:HGR}

In this subsection, we show that the second largest singular value of the DTM is precisely the same as the R$\acute{\mbox{e}}$nyi maximal correlation between random variables $X$ and $Y$, where the marginal distributions $P_X$ and $P_Y$ are the given input and output distributions in the linear information coupling problem \eqref{eq:LICP_single-letter}, and the transition probability kernel $P_{Y|X}$ is the channel $W$. Let us begin with the following definition.

\begin{definition}{\cite{A59-2}}
The \emph{R$\acute{\mbox{e}}$nyi maximal correlation} $\rho_m(X,Y)$ between two random variables $X$ and $Y$ is defined by 
\begin{align}\label{eq:HGR}
\rho_m(X,Y) = \sup E \left[ f(X) g(Y) \right],
\end{align}
where the supremum is over all Borel-measurable functions $f$ and $g$ such that
\begin{align}\label{eq:HGR_con}
\begin{array}{ll}
&E \left[ f(X) \right] = E \left[ g(Y) \right] = 0, \\ 
&E \left[ f^2(X) \right] = E \left[ g^2(Y) \right] = 1.
\end{array}
\end{align}
\end{definition}

The R$\acute{\mbox{e}}$nyi maximal correlation is a measure of dependence of random variables that is stronger and more general than the correlation coefficient, since it allows arbitrary zero-mean, unit-variance functions of $X$ and $Y$. The R$\acute{\mbox{e}}$nyi maximal correlation is first introduced by Hirschfeld \cite{H35} and Gebelein \cite{H41} for discrete random variables and absolutely continuous random variables. R$\acute{\mbox{e}}$nyi \cite{A59,A59-2} compared the R$\acute{\mbox{e}}$nyi maximal correlation to other measures of dependence, and provided sufficient conditions for which the supremum of \eqref{eq:HGR} is achieved. In particular, for discrete random variables $X$ and $Y$, the sufficient conditions are met, and R$\acute{\mbox{e}}$nyi maximal correlation can be attained. Moreover, R$\acute{\mbox{e}}$nyi showed that if the function pair $(\hat{f},\hat{g})$ achieves \eqref{eq:HGR}, then 
\begin{align}\label{eq:Renyi}
\begin{array}{ll}
&E \left[ \hat{g} (Y)  | X \right] = \rho_m(X,Y) \hat{f} (X), \\
&E \left[ \hat{f} (X)  | Y \right] = \rho_m(X,Y) \hat{g} (Y).
\end{array}
\end{align}

Now to see the connection between DTM and the above results, we write 
${\cal F}_X$ and ${\cal G}_Y$ as the spaces of real-valued functions on ${\cal X}$ and ${\cal Y}$, resp; and consider the {\em conditional expectation operator} $E \left[ \cdot \right | X ]: {\cal G}_Y \mapsto {\cal F}_X$ as a map that takes a function of $y$ to a function of $x$. We use $\langle \cdot , \cdot \rangle_{P_X}$ and $\langle \cdot , \cdot \rangle_{P_Y}$  as inner products on ${\cal F}_X$ and ${\cal G}_Y$. This is convenient, as for example \eqref{eq:HGR_con} can be written as 
\begin{align*}
E[f(X)] &= \langle f, {\bold 1}_X\rangle_{P_X} =0,\\
E[g(Y)] &= \langle g, {\bold 1}_Y\rangle_{P_Y} =0,\\
E[f^2(X)] & =\langle f, f\rangle_{P_X} = \norm{f}^2 =1, \\
E[g^2(Y)] &= \langle g, g\rangle_{P_Y} = \norm{g}^2 =1,
\end{align*}
which are simple orthogonality and norm constraints. In the above, we used the notation ${\bold 1}_X$ and ${\bold 1}_Y$ as the constant $1$ functions. 

We also define functions

\begin{align} \label{eq:ONB}
\begin{array}{ll}
&\left\{ \phi_{x'}(\cdot) = \sqrt{P_X(\cdot)}^{-1} \cdot \delta_{x'}(\cdot) , \ x' \in {\cal X} \right\} \\  
&\left\{ \varphi_{y'}(\cdot) = \sqrt{P_Y(\cdot)}^{-1} \cdot \delta_{y'}(\cdot) , \ y' \in {\cal Y} \right\},
\end{array}
\end{align}
where
\[
\delta_{x'} (x) = \left\{
\begin{array}{ll}
1, \ \mbox{if} \  x= x' \\
0, \ \mbox{otherwise}
\end{array}
\right. ,\
\delta_{y' }(y) = \left\{
\begin{array}{ll}
1, \ \mbox{if} \ y = y' \\
0, \ \mbox{otherwise}
\end{array}
\right. ,
\]

It can be verified that these two groups of functions in \eqref{eq:ONB} all have unit norm, and are orthogonal within each group. Thus, they form orthonormal bases of ${\cal F}_X$ and ${\cal G}_Y$ respectively. We now can write the conditional expectation operator $E[\cdot |X]$ in matrix form, with respect to these bases. To do that, consider the $(x',y')$ entry
\begin{align*}
& \langle \phi_{x'} , E \left[ \varphi_{y'} | X = x \right] \rangle_{P_X} \\
=& \sum_{x \in \cX} P_X(x) \cdot \left(  \sqrt{P_X(x)}^{-1} \delta_{x'}(x) \right)\cdot E \left[ \varphi_{y'} | X = x \right]\\
=& \sqrt{P_X(x')}\cdot \left( \sum_{y \in \cY} \sqrt{P_Y(y)}^{-1} \delta_{y'}(y)P_{Y|X} \left( y | x' \right) \right) \\
=& \sqrt{P_Y(y')}^{-1} P_{Y|X} \left( y' | x' \right) \sqrt{P_X(x')}.
\end{align*}
Therefore, this matrix is precise the DTM as we defined. Repeating the same derivation for the operator $E[\cdot |Y]$ reveals that the two operators $E[\cdot |X]$ and $E[\cdot|Y]$ are indeed conjugates of each other. Furthermore, there is a one-to-one correspondence between the singular vectors of the DTM and the singular functions of the conditioned expectation operator. In particular, note that the all-$1$ function $\bold{1}_X$ and $\bold{1}_Y$ has the nice properties
$E \left[ \bold{1}_Y | X \right] = \bold{1}$, and $E\left[\bold{1}_X|Y\right] = \bold{1}_Y$. Thus, $\bold{1}_X$ and $\bold{1}_Y$ are a pair of input and output singular functions of the conditional expectation operator, with singular vector 1. This corresponds to the first singular vector $\left[ \sqrt{P_X(x)}, x \in \cX \right]^T$ of the DTM. In addition, note that the zero-mean constraint~\eqref{eq:HGR_con} on functions in ${\cal F}_X$ is equivalent to the orthogonality to $\bold{1}_X$. Therefore, the rest singular functions of the conditioned expectation operator satisfy~\eqref{eq:HGR_con}, and have a one-to-one correspondence to the singular vectors of the DTM other than the first one.




From \eqref{eq:Renyi}, we know that the pair of function $(\hat{f}, \hat{y})$ that achieves the maximal correlation must be a pair of input/output singular vectors of the conditional expectation operator. The fact that such correlation $|\rho_m(X,Y)|$ is no larger than $1$ quantifies the contraction behavior of the operators, and is a manifestation of the data processing inequality. We summarize in the following proposition.


\begin{prop}\label{prop:DTM_max}
The second largest singular value of the DTM is the R$\acute{\mbox{e}}$nyi maximal correlation $\rho_m(X,Y)$. Moreover, the functions $\hat{f}$ and $\hat{g}$ maximizing \eqref{eq:HGR} can be obtained from the left and right singular vectors of the DTM with the second largest singular value, corresponding to the orthonormal bases \eqref{eq:ONB}.
\end{prop}

The relation between \eqref{eq:LICP_single-letter} and the R$\acute{\mbox{e}}$nyi maximal correlation $\rho_m(X,Y)$ was also shown in~\cite{ET98} and~\cite{Witsenhausen75} by different approaches. The value of this connection, to our purpose, is that it provides yet another view of the local approximation approach. More interestingly, 
if we recall that R$\acute{\mbox{e}}$nyi's original work on maximal correlation was to characterize the dependence between two random variables; now the spectrum analysis of the DTM, or the conditional expectation operator, provides us with a broader range of quantities that might be of interests. For example, now we can ask not only about the maximal correlation, but also ``second" maximal correlation, which corresponds to the largest two singular values of the DTM. In the following sections, we will generalize the concept to multi-terminal cases, which potentially leads to concepts such as the maximal correlation between more than two random variables. As we will show, this generalization is indeed quite simple if we keep the local approximations.  



\section{The General Broadcast Channels} \label{sec:broadcast}

In this section, we apply the local approximation approach to general broadcast channels, and study the corresponding linear information coupling problems. We first illustrate our technique by considering the $2$-user broadcast channel, and then the extension to the $K$-user case is straight forward. 

A $2$-user general broadcast channel with input $X \in \cX$, and outputs $Y_1 \in \cY_1$, $Y_2 \in \cY_2$, is specified by the memoryless channel matrices $W_1$ and $W_2$. These channel matrices specify the conditional distributions of the output signals at two users, $1$ and $2$, as $W_i ( y_i | x ) = P_{Y_i|X} ( y_i | x )$, for $i = 1,2$. Let $M_1$, $M_2$, and $M_0$ be the two private messages and the common message, with rate $R_1$, $R_2$, and $R_0$, respectively. Then, using Fano's inequality, the multi-letter capacity region of the general broadcast channel is the set of rate tuple $(R_0, R_1, R_2)$ such that 
\begin{equation} \label{eq:private1}
\left\{
\begin{array}{clr}
R_0 &\leq \frac{1}{n} \min \{ I(U ; \uY_1) , I(U ; \uY_2) \}, \\ 
R_1 &\leq \frac{1}{n} I(V_1 ; \uY_1), \\
R_2 &\leq \frac{1}{n} I(V_2 ; \uY_2),
\end{array} \right.
\end{equation}
for some mutually independent random variables $U$, $V_1$, and $V_2$, such that $( U , V_1 , V_2 ) \rightarrow \uX \rightarrow \uY_1$ and $( U , V_1 , V_2 ) \rightarrow \uX \rightarrow \uY_2$, are both Markov chains. 
The signal vectors here all have the same dimension $n$. In principle, one should just optimize this rate region by finding the optimal coding distributions. However, since $n$ can potentially be arbitrarily large, finding the structure of these optimal input distributions is necessary.

Now, we want to apply the local approximation technique we developed in section \ref{sec:P2P} to this broadcast channel problem. As a natural generalization from the point-to-point channel case (\ref{eq:LICP_multi-letter}), the linear information coupling problem of this $2$-user broadcast channel is the characterization of the rate region:
\begin{align}\label{LICP-BC1}
&\left\{
\begin{array}{ll}
R_0 \leq \frac{1}{n} \min \{ I(U ; \uY_1) , I(U ; \uY_2) \} \\
R_1 \leq \frac{1}{n} I(V_1 ; \uY_1) \\
R_2 \leq \frac{1}{n} I(V_2 ; \uY_2) 
\end{array}\right. \\ \notag
\mbox{subject to:} \quad &
( U , V_1 , V_2 ) \rightarrow \uX \rightarrow (\uY_1, \uY_2), \\ \notag 
&\frac{1}{n} I( U , V_1 , V_2 ; \uX ) \leq \frac{1}{2} \epsilon^2, \\ \notag
& \frac{1}{n} \| P_{\uX | (U,V_1,V_2) = (u,v_1,v_2)} - P_{\uX} \|^2 = O(\epsilon^2), \\ \notag
& \forall \ (u,v_1,v_2) \in {\cal U} \times {\cal V}_1 \times {\cal V}_2,
\end{align}
where $U , V_1 , V_2$ are mutually independent random variables.

This rate region is the same as the capacity region (\ref{eq:private1}) except for the local constraints. The operational meaning of these constraints are similar to that for the point-to-point case. That is, we consider modulating all the common and private messages entirely as a thin layer of information into the input symbol sequence $\uX$. 



The first simplifying observation is that the characterization of (\ref{LICP-BC1}) involves the optimization over multiple rates $R_0$, $R_1$, and $R_2$, with respect to different messages $M_0$, $M_1$, and $M_2$. This can indeed be separated into three sub-problems. The idea here is that, while the conditional distribution $P_{\uX|( U , V_1 , V_2 ) = ( u , v_1 , v_2 )}$ is perturbed from $P_{\uX}$ by some vector $J_{u , v_1 , v_2}$ that is in general a joint function of $U$, $V_1$, and $V_2$, by the first order approximation, it is enough to only consider perturbation vectors that can be written as the linear combination of three vectors $J_{u}$, $J_{v_1}$, and $J_{v_2}$. This fact is shown in the following Lemma.


\begin{lemma} \label{prop:BC}
The rate region \eqref{LICP-BC1} is, up to the first order approximation, the same as the following rate region with the constraints separated for $U,V_1$, and $V_2$:
\begin{align}\label{LICP-BC2}
&\left\{
\begin{array}{ll}
R_0 \leq \frac{1}{n} \min \{ I(U ; \uY_1) , I(U ; \uY_2) \} \\
R_1 \leq \frac{1}{n} I(V_1 ; \uY_1) \\
R_2 \leq \frac{1}{n} I(V_2 ; \uY_2) 
\end{array}\right. \\ \notag
\mbox{subject to:} \ &
( U , V_1 , V_2 ) \rightarrow \uX \rightarrow (\uY_1, \uY_2), \frac{1}{n} I(U ; \uX) \leq \frac{1}{2} \epsilon^2_0, \\ \notag
&  \frac{1}{n} I(V_1 ; \uX) \leq \frac{1}{2} \epsilon^2_1, \frac{1}{n} I(V_2 ; \uX) \leq \frac{1}{2} \epsilon^2_2,\ \sum_{i=1}^{3} \epsilon_i^2 = \epsilon^2 \\ \notag
& \frac{1}{n} \| P_{\uX | (U,V_1,V_2) = (u,v_1,v_2)} - P_{\uX} \|^2 = O(\epsilon^2), \\ \notag
& \forall \ (u,v_1,v_2) \in {\cal U} \times {\cal V}_1 \times {\cal V}_2,
\end{align}
\end{lemma}

\begin{proof}
Appendix~\ref{app:lem3}.
\end{proof}


Now for a tuple of $(\epsilon_0, \epsilon_1, \epsilon_2)$ with $\sum_{i=1}^{3} \epsilon_i^2 = \epsilon^2$,  the optimization problem \eqref{LICP-BC2} reduces to three sub-problems: for $i = 1,2$, the optimization problems for transmitting the private messages $M_i$
\begin{align}
\label{eq:BC_private}
\max. & \quad \frac{1}{n}  I(V_i ; \uY_i)\\ \notag
\mbox{subject to:} & \quad V_i \rightarrow \uX \rightarrow \uY_i , \frac{1}{n} I(V_i ; \uX) \leq \frac{1}{2} \epsilon^2_i,\\ \notag
& \quad \frac{1}{n} \| P_{\uX | V_i = v_i} - P_{\uX} \|^2 = O(\epsilon_i^2), \ \forall \ v_i,
\end{align}
and the optimization problem for the common message $M_0$
\begin{align}
\label{eq:BC_single_public}
\max. & \quad \frac{1}{n} \min \left\{ I(U ; \uY_1) , I(U ; \uY_2) \right\} \\ \notag
\mbox{subject to:} & \quad U \rightarrow \uX \rightarrow (\uY_1, \uY_2) , \frac{1}{n} I(U ; \uX) \leq \frac{1}{2} \epsilon^2_0,\\ \notag
& \quad \frac{1}{n} \| P_{\uX | U=u} - P_{\uX} \|^2 = O(\epsilon_0^2), \ \forall \ u.
\end{align}

As in the point-to-point channel case, we assume that the input distribution of $\uX$, as the operating point, is i.i.d. $P_X$. Hence, the output distributions of the two outputs $Y_1$ and $Y_2$ are also i.i.d. $P_{Y_1}$ and $P_{Y_2}$. The conditional distributions is denoted as perturbations from the marginal distribution, which are written as $P_{\uX|U = u} = P_X^{(n)} + \sqrt{n} \epsilon_0 \cdot J_{u} $ and $P_{\uX|V_i = v_i} = P_X^{(n)} + \sqrt{n} \epsilon_i \cdot J_{v_i} $ for $i = 1,2$.

Then, the optimization problems (\ref{eq:BC_private}) for private messages are the same as the linear information coupling problem for the point-to-point channel (\ref{eq:LICP_multi-letter}). Thus, by defining the single-letter DTM's $B_i \triangleq \left[\sqrt{P_{Y_i}}^{-1}\right] W_i \left[\sqrt{P_X}\right]$ for $i = 1,2$, we can solve (\ref{eq:BC_private}) with the same procedure as (\ref{eq:LICP_multi-letter}), and the single-letter solutions are optimal. 

The optimization problem (\ref{eq:BC_single_public}) is, however, fundamentally different from the other two. Suppose that the weighted perturbation vector $\K_{u} = \left[\sqrt{P^{(n)}_{X}}^{-1}\right]J_{u}$, the local version of the problem can be written as
\begin{align}\label{eq:abcd}
\max_{\K_u} \min \left\{ \sum_{u} P_U(u) \left\| B_1^{(n)} \K_{u}  \right\|^2 , \sum_{u} P_U(u) \left\| B_2^{(n)} \K_{u}  \right\|^2 \right\},
\end{align}
subject to
\begin{align}\notag
\sum_{u} P_U(u) \cdot \| \K_u \|^2 = 1,
\end{align}
and also the constraints~\eqref{eq:bbb_2} and~\eqref{eq:bbb_3} that guarantee the validity of the weighted perturbation vector. Here, the $B_i^{(n)}$ is the $n^{th}$ Kronecker product of the single-letter DTM $B_i$, for $i = 1,2$. 

Our goal here is to check whether there exists a single-letter optimal solution  for the problem of transmitting the common message \eqref{eq:abcd}. Similar to the process of analyzing the point-to-point problem, this is carried out in several steps: 1) we need to check whether the optimal solution in \eqref{eq:abcd}, perturbation vectors $\K_u^*$, leads to $P_{\vec{x}|U=u}$ that take product form; 2) we hope to find a time-sharing/convexity argument to show that the optimal choices of $P_{\vec{x}|U=u}$ are not only independent, but also identical from letter to letter; 3) we need to have control over the  cardinality of $U$. In particular, we would hope that the cardinality of $U$ does not change with $n$, and should not increase with $K$, the number of receivers of the broadcast channel. 

Before answering these questions in a formal statement, we would first make some intuitive discussions, to point out the key differences between the broadcast channel and the point-to-point case. 

The key difference between the optimization problem \eqref{eq:abcd} and its counterpart for the point-to-point channel \eqref{eq:bbb200} is that we want to design a set of perturbation vectors, whose images through two separate linear maps, $B^{(n)}_1$ and $B^{(n)}_2$, are large at the same time. This involves the tradeoff between two linear systems, and is exactly the key issue in broadcasting common messages. 

While generally the tradeoff between two linear systems can be a rather messy problem, the special structure of $B^{(n)}_1$ and $B^{(n)}_2$ turns out to be quite useful. Since both $i=1,2$, $B^{(n)}_i$ is the tensor product of the corresponding single-letter DTM's $B_i$, the singular vectors of $B^{(n)}_i$ also take the form of tensor products of the singular vectors of $B_i$. That is, if $\vec{v}_{i,0}, \vec{v}_{i,1}, \ldots, \vec{v}_{i, M}$ are the singular vectors of $B_i$, with the corresponding singular values $\sigma_{i,0} \geq \sigma_{i,1}\geq \ldots\geq \sigma_{i,M}$, then for any $j_1, j_2, \ldots, j_n \in \{0, \ldots, M \}$ 

\begin{align}
\label{eqn:singularvector}
\vec{v}_{i, j_1} \otimes \vec{v}_{i, j_2} \otimes \ldots \otimes \vec{v}_{i, j_n}
\end{align}
is a singular vector for $B^{(n)}_i$, with the corresponding singular value $\sigma_{i, j_1} \cdot \sigma_{i, j_2} \ldots \sigma_{i, j_n}$. For each $i \in \{1,2\}$, the corresponding collection of singular vectors form an orthonormal basis on the space of joint distributions over $X_1, \ldots, X_n$. 

Moreover, $B_1$ and $B_2$ share the same largest singular value $\sigma_{i,0}=1$ and the same corresponding singular vector $\vec{v}_{i, 0} = [\sqrt{P_{X}(x)}, x \in {\cal X}]^T$. This means that both linear systems $B_1^{(n)}$ and $B_2^{(n)}$ output larger images for those singular vectors of the form \eqref{eqn:singularvector} with smaller indices $j_1, \ldots, j_n$. Note that the vector with $j_1=j_2=\ldots=j_n=0$ is an invalid choice of perturbation vector, thus a direction with large output images through both systems should have all but one of the indices equal to $0$. Put it another way, if we pick a perturbation vector that has a non-zero image in the form of \eqref{eqn:singularvector}, with any choice of $j_1, \ldots, j_n$, then replacing all but one indices by $0$ would result in a larger image through both $B_1^{(n)}$ and $B_2^{(n)}$. Formalizing this argument, we can conclude that the optimizer of \eqref{eq:abcd} must be $\K_u$'s that are linear combinations of these singular vectors. This corresponds precisely to a resulting $P_{\vec{X}|U=u}$ in the product form, i.e., $X_1, X_2, \ldots, X_n$ are independent conditioned on $U$, just like in the point-to-point case. 

The main difference between the broadcasting channel and the point-to-point channel comes with the choice of $P_U$. Recall previously, in \eqref{eq:opt_prob}, we used the fact that for any choice of the direction of $\K$, there is a linear relation between the constraint on $\norm{\K}^2$ and the objective $\norm{B\K}^2$ to conclude that one can without loss of generality pick $P_U$ to be the binary uniform distribution, and focus only on choosing a single direction of $\K$ that maximizes the ratio $\norm{B\K}^2/\norm{\K}^2$. In the broadcast channel, however, since our objective function is the minimum output image from multiple linear systems, the above argument no longer holds. In fact, in some cases, it can be beneficial to choose multiple $\K_u$'s with different directions, and average over the output squared norms, as without the averaging, some of such choices might be eliminated by the $\min\{\cdot, \cdot\}$ operation. In fact, when we have $K>2$ receivers, the objective function becomes the minimum among the outputs of $K$ linear systems, and the number of $\K_u$'s we need to achieve the optimum increases with $K$. This is summarized in the following Theorem and example.

\begin{thm} \label{thm:BC}
Let $B_i$ be the DTM of some DMC with respect to the same input distributions $P_X$, for $i = 1,2, \ldots , K$, let $\uv_0 = [\sqrt{P_X(x)}, x\in {\cal X}]^T$ be the common singular vector of $B_i$ with the largest singular value of $1$, and  $\uv_0^{(n)} = \uv_0 \otimes \uv_0, \ldots, \uv_0$ be the $n$-th Kronecker product, then for the linear information coupling problem 
\begin{align} \label{eq:k-BC}
\lambda^{(n)} = \max \min_{1 \leq i \leq K} \left\{ \sum_u P_U(u) \cdot \left\| B_i^{(n)} \K_{u}  \right\|^2 \right\},
\end{align}
where the maximization is taken over all $P_U(\cdot) $ and $\K_u$ such that 
\begin{align*}
&\sum_u  P_U(u) \cdot \| \K_{u} \|^2 = 1\\
&\langle \K_u, \uv_0^{(n)}  \rangle =0, \quad \forall u
\end{align*}
we have
\begin{itemize}
\item[a)] There exists an optimal choice where $\K_u$ for all $u$ take product form: 
\begin{align*}
\K_u = \uv^{(1)} \otimes \uv^{(2)} \otimes \ldots \otimes \uv^{(n)}
\end{align*}
where all but one of the $\uv^{(i)}$'s are equal to $\uv_0$. 
\item[b)] $\lambda^{(n)}= \lambda^{(1)}$ with the cardinality of $U$ no larger than $K$
\item[c)] With the extra constraint that $U$ is binary and the two corresponding $\K_u$'s lie in the same direction, the problem becomes 
\begin{align*}
\lambda_B^{(n)} = \max_{\K: \norm{\K}^2=1, \langle \K, \vec{v}_0\rangle =0} \min_{1\leq i \leq K} \left\{ \norm{B_i^{(n)} \K}^2\right\}
\end{align*} 
We have 
\begin{align*}
\sup_n \lambda_B^{(n)} = 
\left\{
\begin{array}{ll}
\lambda_B^{(1)} \ \mbox{if $k = 2$},\\
\lambda_B^{(k)} \ \mbox{if $k > 2$},
\end{array}\right.
\end{align*}
and for $k>2$, there exists $k$-user broadcast channels such that
\begin{align*}
\lambda^{(k)} > \lambda^{(k-1)}.
\end{align*}
\end{itemize}
In other words, with a constraint on the cardinality of $U$, single-letter solutions are not optimal in general. However, there always exists a  $K$-letter optimal solution. 
\end{thm}


\begin{proof}
See Appendix \ref{ap:proof 1}.
\end{proof}

The following example illustrates that when there are more than $2$ receivers, i.i.d. distributions simply do not have enough degrees of freedom to be optimal in the tradeoff of the $K$ linear systems. Therefore, one has to design multi-letter product distributions to achieve the optimal. 
The following example, constructed with the geometric method, illustrates the key ideas. 

\begin{figure}
\centering 
\subfigure[]{
\def\svgwidth{.5\columnwidth}
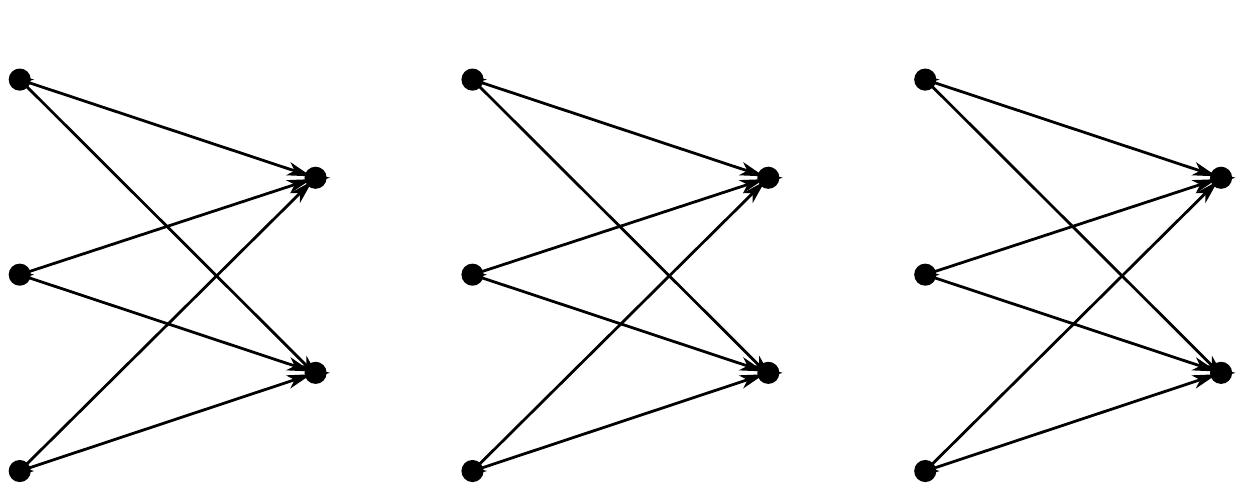
\label{fig:Windmill_Channel}
}
\subfigure[]{
\def\svgwidth{0.35\columnwidth}
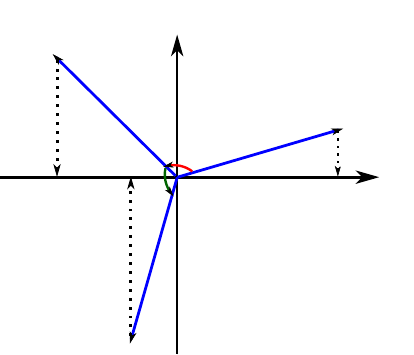
\label{fig:windmill_1}
}
\subfigure[]{
\def\svgwidth{0.25\columnwidth}
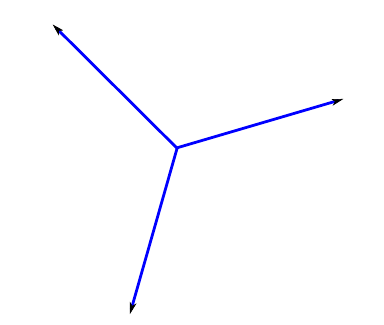
\label{fig:Windmill}
}
\caption{(a) A ternary input broadcast channel. (b) The lengths of the output images of $\K$ through matrices $B_1$, $B_1$, and $B_1$ are the scaled projection of $\K$ and its rotated version to the horizontal axis. (c) The optimal perturbations over $3$ time slots.}
\end{figure}

\begin{example}[The windmill channel]
We consider a $3$-user broadcast channel as shown in Figure \ref{fig:Windmill_Channel}. The input alphabet $\cX$ is ternary, so that the perturbation vectors have $2$ dimensions and can be easily visualized. Suppose that  $P_X$ is fixed as $[\frac{1}{3} \ \frac{1}{3} \ \frac{1}{3}]^T$, then the DTM $B_1$ for the first receiver $Y_1$ is
\begin{align*}
B_1
 = \sqrt{\frac{2}{3}}
\begin{bmatrix}                
  \frac{1}{2} & 1 - \delta & \delta \\
  \frac{1}{2} & \delta & 1 - \delta
\end{bmatrix}
=\sigma_0 \uu_0 \uv_0^T + \sigma_1 \uu_1 \uv_1^T,
\end{align*}
where $\sigma_0 = 1$ and $\sigma_1 = \sqrt{\frac{2}{3}} (1 - 2 \delta)$ are the singular values of $B_1$, and $\uu_0 = [\frac{1}{\sqrt{2}} \ \frac{1}{\sqrt{2}}]^T$, $\uu_1 = [\frac{1}{\sqrt{2}} \ \frac{-1}{\sqrt{2}}]^T$, $\uv_0 = [\frac{1}{\sqrt{3}} \ \frac{1}{\sqrt{3}} \ \frac{1}{\sqrt{3}}]^T$, and $\uv_1 = [0 \ \frac{1}{\sqrt{2}} \ \frac{-1}{\sqrt{2}}]^T$ are the corresponding left and right singular vectors. 

Intuitively, one can think of this channel from the transmitter to the first receiver as a projection operator. It takes the input variation and write it with respect to an orthonormal basis $\uv_1$ and $\uv_2 = [\frac{2}{\sqrt{6}} \ \frac{-1}{\sqrt{6}} \ \frac{-1}{\sqrt{6}}]^T$, takes the first element along $\uv_1$, scales  by $\sigma_1$, and maps to the output direction of $\uu_1$. The input variation in the direction of $\uv_2$ is wiped out by the channel, and has no effect to the output distribution. 

Similarly, we can write the  the DTM's $B_2$ and $B_3$ for receivers $Y_2$ and $Y_3$  as
\begin{align*}
B_2
= \sqrt{\frac{2}{3}}
\begin{bmatrix}                
  \delta & \frac{1}{2} & 1 - \delta  \\
  1 - \delta & \frac{1}{2} & \delta
\end{bmatrix}
=\sigma_0 \uu_0 \uv_0^T + \sigma_1 \uu_1
[1 \ 0]
\begin{bmatrix}                
  \cos \frac{2\pi}{3} & -\sin \frac{2\pi}{3} \\
  \sin \frac{2\pi}{3} & \cos \frac{2\pi}{3}
\end{bmatrix}
\begin{bmatrix}                
  \underline{v}_1^T \\
  \underline{v}_2^T
\end{bmatrix},
\end{align*}
and
\begin{align*}
B_3
= \sqrt{\frac{2}{3}}
\begin{bmatrix}                
  1 - \delta & \delta & \frac{1}{2}   \\
  \delta & 1 - \delta & \frac{1}{2} 
\end{bmatrix} 
=\sigma_0 \uu_0 \uv_0^T + \sigma_1 \uu_1
[1 \ 0]
\begin{bmatrix}                
  \cos \frac{4\pi}{3} & -\sin \frac{4\pi}{3} \\
  \sin \frac{4\pi}{3} & \cos \frac{4\pi}{3}
\end{bmatrix}
\begin{bmatrix}                
  \underline{v}_1^T \\
  \underline{v}_2^T
\end{bmatrix},
\end{align*}

The channel $B_2$ also represents the input variation with the orthonormal basis $\uv_1, \uv_2$. However, it rotates the 2-dimensional coefficient vector by $\frac{2\pi}{3}$, before taking the first element, and similarly scales by $\sigma_1$ and maps to $\uu_1$. The channel $B_3$ does the same thing, except the rotation is $\frac{4\pi}{3}$. This is illustrated in Figure \ref{fig:windmill_1}, from which the name ``windmill" channel should be obvious.



Now if we pick a single input direction $\K$, it can be shown that for any $\K$ with $\| \K \|^2 = 1$, $\min \left\{ \| B_1 \K \|^2, \| B_2 \K \|^2, \| B_3 \K \|^2 \right\} \leq \frac{1}{4} \sigma_1^2$. To see that, for example if we simply take $\K= \uv_1$. The output squared norm $\norm{B_1 \K}^2$ is $\sigma_1^2$, but $\norm{B_2 \K}^2= \norm{B_3\K}^2 = \frac{1}{4} \sigma_1^2$. A variation from this choice would reduce either $\norm{B_2 \K}$ or $\norm{B_3 \K}$, hence further reduce the minimum. The point here is that the largest output $\norm{B_1 \K}$ is not used. 

For a better choice, we take ${\cal U}= \{0,1,2\}$, with $P_U(u) = 1/3$ for all $u$.  Write 
\begin{align*}
\phi_\theta = \cos \theta \cdot \uv_1 + \sin \theta \cdot \uv_2
\end{align*}
we pick for any $\theta_0$
\begin{align*}
\K_{0} = \phi_{\theta_0}, \qquad \K_{1} = \phi_{\theta_0+ \frac{2\pi}{3}}, \qquad \K_{2} = \phi_{\theta_0+ \frac{4\pi}{3}}
\end{align*}
This corresponds to choosing the conditional distributions $P_{X|U=u}(x) = P_X (x) + \epsilon \sqrt{P_X(x)} \K_u(x), \forall x,u$. It is easy to verify that 
\begin{align*}
\sum_u P_U(u) \norm{B_i \cdot \K_u}^2 = \frac{1}{2} \sigma_1^2
\end{align*}
for all $i=1,2,3$ and regardless of the value of $\theta_0$. 

Equivalently, we can turn this averaging over $U$ into an average over time. To do that, we consider a 3-letter solution to the problem, where we can keep $U$ to be binary, and set
\begin{align*}
P_{X^3|U} 
= (P_X \pm \epsilon \sqrt{P_X} \K_0) \otimes  (P_X \pm \epsilon \sqrt{P_X} \K_1) \otimes (P_X \pm \epsilon \sqrt{P_X} \K_2)  
\end{align*}
where the signs depend on the value of $U$. This corresponds to perturbations along the vector 
\begin{align*}
\K_0 \otimes \uv_0 \otimes \uv_0 + \uv_0 \otimes \K_1 \otimes \uv_0 + \uv_0 \otimes \uv_0 \otimes \K_2
\end{align*}
The point is that with a $K=3$ letter solution, we can limit the cardinality of $U$ to be not increasing with the number of receivers $K$.


Translating this solution to the coding language, the later solution can be interpretted as repeating the binar message $U$ in three time slots. This can be thought as feeding the common information in turn to three individual recievers. 

\end{example}

\begin{remark}
The above analysis can be generalized to general broadcast channels with $K$ receivers. In such problems, there are $2^K-1$ different types of messages, one to be decoded by a particular subset of $K' \leq K$ receivers. With the same argument that we separated the design of the private messages from the common message in Lemma \ref{prop:BC}, we can without loss of the optimality separately design the perturbation vectors for each of such messages, which is equivalent as a $K'$ receiver broadcast channel. The overall transmitted codeword is then the superposition of all such messages. It is then modulated on a sequence of i.i.d. $P_X$ symbols as described in section \ref{sec:CALC}. Each receiver will decode all the messages designated to him. Since all perturbations are local, even the order of decoding these messages can be arbitrary. 

This scheme above is by no means designed to achieve the capacity region. The local approximations are so crude that some critical issues in achieving the capacity, such as the order of decoding different messages, have no effect on the approximated performance, and thus can not be addressed with this approach. Moreover, some of observations, such as single-letter solutions are not generally optimal for $K>2$ receiver channels, have been reported \cite{CA09}, with more general terms. The geometric analysis, however, does help to reveal some issues in a very explicit way, and suggests new directions of designing coding schemes that might be even applicable to the non-local problems, such as using a larger (but finite) cardinality of $U$ to balance the performance between multiple users.  
\end{remark}

\section{The Multiple Access Channels with Common Sources} \label{sec:multiple}

In this section, we want to apply the local geometric approach to the multiple access channel (MAC), and study the corresponding linear information coupling problem. Specifically, we consider the set-up, where the transmitters can not only have the knowledge of their own private sources, but each subset of transmitters also share the knowledge of a common source. In particular, all these private and common sources are assumed to be independent with each other. 

The MAC with common sources is a celebrated information theory problem~\cite{TJ91},~\cite{TAM80}. The main challenge of investigating efficient transmission schemes for common messages lies on modeling the benefit of the collaboration between transmitters that shares common knowledge. This collaboration gain is well studied as the beam-forming gain for Gaussian additive channels; however, for general discrete memoryless MACs, there still lacks a systematic and simple approach to quantitively compute this gain. In this section, we aim to provide new perspective on understanding the transmitter collaboration gain via our local approach. To illustrate how our technique is applied to this problem, let us first consider the $2$-transmitter MAC with the common source.

Suppose that the $2$-transmitter multiple access channel has the inputs $X_1 \in \cX_1$, $X_2 \in \cX_2$, and the output $Y \in \cY$. The memoryless channel is specified by the channel matrix $W$, where $W (y | x_1, x_2 ) = P_{Y|X_1,X_2} (y|x_1,x_2)$ is the conditional distribution of the output signals. We want to communicate three messages $M_1$, $M_2$, and $M_0$ to the receiver $Y$ with rate $R_1$, $R_2$, and $R_0$, where $M_1$ and $M_2$ are privately observed by transmitters $1$ and $2$, respectively, and both transmitters have the common knowledge on $M_0$. Then, following the same arguments as the broadcast channel in section~\ref{sec:broadcast}, the single-letter version of the linear information coupling problem of the MAC is formulated as three sub-problems: the optimization problems for the private sources $M_1$ and $M_2$
\begin{align}
\label{eq:MAC_private}
\max. & \quad \frac{1}{n}  I(V_i ; Y)\\ \notag
\mbox{subject to:} & \quad V_i \rightarrow X_i \rightarrow Y , \frac{1}{n} I(V_i ; X_i) \leq \frac{1}{2} \epsilon^2_i, \\ \notag
& \quad \frac{1}{n} \| P_{X_i | V_i = v_i} - P_{X_i} \|^2 = O(\epsilon_i^2), \ \forall \ v_i.
\end{align}
for $i = 1,2$, and the optimization problem for the common source $M_0$
\begin{align}
\label{eq:MAC_single_public}
\max. & \quad \frac{1}{n}  I(U ; Y) \\ \notag
\mbox{subject to:} & \quad U \rightarrow (X_1, X_2) \rightarrow Y , \frac{1}{n} I(U ; X_1, X_2) \leq \frac{1}{2} \epsilon^2_0, \\ \notag
&  \frac{1}{n} \| P_{X_1, X_2 | U = u } - P_{X_1, X_2} \|^2 = O(\epsilon_0^2), \ \forall \ u.
\end{align}
Now, let us employ the notations $P_{X_i|V_i=v_i} = P_{X_i} +  \epsilon_i \cdot J_{v_i} $ and $P_{X_i|U=u} = P_{X_i} +  \epsilon \cdot J_{i,u} $, for $i = 1,2$, and let $\K$ denote the scaled perturbation vectors. Then, note that $P_{X_1, X_2 | U  } = P_{X_1|U} \cdot P_{X_2|U}$ since $U$ is the only common message shared by $X_1$ and $X_2$, the problems~\eqref{eq:MAC_private} and~\eqref{eq:MAC_single_public} can be simplified to local problems
\begin{align} \label{eq:local_private_MAC}
\max. \quad & \sum_{v_i} P_{V_i} (v_i) \cdot \left\| B_{i} \K_{v_i} \right\|^2 \\ \notag
\mbox{subject to:} \quad & \sum_{v_i} P_{V_i} (v_i) \cdot  \left\| \K_{v_i} \right\|^2 = 1,
\end{align}
for $i = 1,2$, and
\begin{align} \label{eq:local_common_MAC}
\max. \quad & \sum_{u} P_{U} (u) \cdot \left\| B_{1} \K_{1,u} + B_2 \K_{2,u} \right\|^2 \\ \notag
\mbox{subject to:} \quad & \sum_{u} P_{U} (u) \cdot \left(  \left\| \K_{1,u} \right\|^2 + \left\| \K_{2,u} \right\|^2 \right) = 1.
\end{align}
Here, the DTM $B_i$ is defined as $B_i \triangleq \left[ \sqrt{P_{Y}}^{-1} \right] W_i \left[ \sqrt{P_{X_i}} \right]$, for $i = 1,2$, with the channel matrix
\begin{align*} 
W_i ( y | x_i ) \triangleq \sum_{x_{3-i} \in \cX_{3-i}} W (y|x_1, x_2) P_{X_{3-i}} (x_{3-i}).
\end{align*}
\begin{remark}
Compare~\eqref{eq:MAC_private} and~\eqref{eq:MAC_single_public}, the main difference is that when optimizing~\eqref{eq:MAC_single_public}, the perturbation vector can be chosen around the joint distribution $P_{X_1, X_2}$. On the other hand, the problem~\eqref{eq:MAC_private} is optimized over a projected space, namely, the perturbation vector is only allowed to be designed around the marginal distribution $P_{X_i}$. Thus, the optimal solution of~\eqref{eq:MAC_single_public} is larger than~\eqref{eq:MAC_private} due to the more freedom of designing the perturbation vectors over a higher dimensional space. Moreover, the solution of~\eqref{eq:MAC_single_public} quantitively illustrates how benefit the collaboration between transmitters is. Therefore, our approach in fact provides a way to visualize the structure the collaboration gain.
\end{remark}

In general, the problems~\eqref{eq:MAC_private} and~\eqref{eq:MAC_single_public} shall be formulated as multi-letter problems. Simplifying to linear algebra problems, the multi-letter version of~\eqref{eq:MAC_private} is precisely the same as the point-to-point problem in section~\ref{sec:P2P}. On the other hand, the multi-letter version of~\eqref{eq:MAC_single_public}, simplified by taking $U$ as Bernoulli(1/2) random variable\footnote{Following the same argument as the point-to-point case, this choice of $U$ is indeed without loss of the optimality.}, can be written as
\begin{align} \label{eq:multi_local_common_MAC}
\max_{ \left\| \K_{1,u} \right\|^2 + \left\| \K_{2,u} \right\|^2 = 1} \left\| B_{1}^{(n)} \cdot \K_{1,u} + B_2^{(n)} \cdot \K_{2,u} \right\|^2.
\end{align} 
In addition, for both $i = 1,2$, $\K_{i,u}$ has to be orthogonal to $\left[ \sqrt{P_{X_i}}, x_i \in \cX_i \right]^T$ to guarantee that $P_{X_i|U}$ is a valid probability distribution. This constraint is slightly stronger than the corresponding one in the point-to-point case. Therefore, the verification of the single-letter optimality of~\eqref{eq:multi_local_common_MAC} is carried out in two steps: 1) we need to show that the second largest singular value of multi-letter linear map $B_{0,n} = [B_1^{(n)} \ B_2^{(n)}]$ is the same as the corresponding single-letter one; 2) we need to demonstrate that the right singular vector of $B_{0,1}$ w.r.t. the second largest singular value satisfies the stronger orthogonality constraint. The following theorem~\ref{thm:MAC} summarizes these two steps for a more general $K$-user case, and concludes that~\eqref{eq:multi_local_common_MAC} is in fact single-letter optimal.

\begin{thm} \label{thm:MAC}
For a multiple access channel with transmitters $1,2, \ldots , k$, let $B_i$ be the corresponding DTM's. Then, the second largest singular value of $B_{0,n} = \left[ B_1^{(n)} \ \ldots \ B_k^{(n)} \right]$ is the same as $B_0 \triangleq B_{0,1}$. Moreover, let $\K_{U=u} = \left[ \K_1^T \ \ldots \ \K_k^T \right]^T$ be the singular vector of $B_0$ with the second largest singular value, where $\K_i$ is an $| \cX_i |$-dimensional vector. Then, $\K_i$ is orthogonal to $\left[ \sqrt{P_{X_i}}, x_i \in \cX_i \right]^T$, for all $1 \leq i \leq k$.
\end{thm}

\begin{example} \label{example:MAC}
Consider the binary adder channel as shown in Figure \ref{fig:example3}, where $X_1$ and $X_2$ are both binary inputs, and the tenary output $Y = X_1 + X_2$ (the arithmetic addition, not modulo $2$). The empirical distribution of both $P_{X_1}$ and $P_{X_2}$ are fixed as $[\frac{1}{2} \ \frac{1}{2}]^T$, and the corresponding output distribution is $[\frac{1}{4} \ \frac{1}{2} \ \frac{1}{4}]^T$. The DTM's for transmitter $1$ and $2$ are
\begin{align*}
B_1 = B_2 = \left[
      \begin{array}{cccc}
        \frac{1}{\sqrt{2}} & 0 \\
        \frac{1}{2} & \frac{1}{2} \\
        0 & \frac{1}{\sqrt{2}} \\
      \end{array}
    \right].
\end{align*}
Thus, the DTM for the common source is
\begin{align*}
B_0 = \left[
      \begin{array}{cccc}
        \frac{1}{\sqrt{2}} & 0 & \frac{1}{\sqrt{2}} & 0 \\
        \frac{1}{2} & \frac{1}{2} & \frac{1}{2} & \frac{1}{2} \\
        0 & \frac{1}{\sqrt{2}} & 0 & \frac{1}{\sqrt{2}} \\
      \end{array}
    \right].
\end{align*}
The second largest singular value of $B_0$ is $1$ with right singular vector $[\frac{1}{2} \ \frac{-1}{2} \ \frac{1}{2} \ \frac{-1}{2}]^T$. In comparison, the second largest singular of both $B_1$ and $B_2$ are $\frac{1}{\sqrt{2}}$ with right singular vector $[\frac{1}{\sqrt{2}} \ \frac{-1}{\sqrt{2}}]^T$. Therefore, there is a $3$dB \emph{coherent combining gain} that arises from the cooperation between transmitters due to their common knowledge.
\end{example}

\begin{figure}
\centering
\def\svgwidth{0.4\columnwidth}
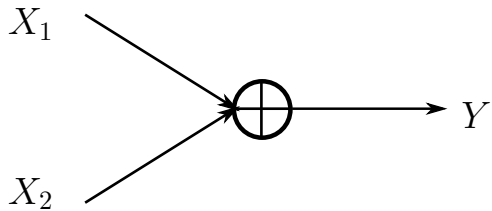
\caption{The binary addition channel with two binary inputs $X_1$ and $X_2$, and the ternary output $Y = X_1 + X_2$.}
\label{fig:example3}
\end{figure}

\begin{remark}
It is straight forward to extend the analysis in this section to $K$-user MAC, in which there are $2^K-1$ different types of sources, and each source is accessible by a particular subset of transmitters. With a similar argument as the broadcast channel, the optimal perturbation vectors for transmitting each type of message are separately designed, and the corresponding subset of transmitters cooperatively modulate that source into the input symbols. In particular, the cooperation between transmitters increases the efficiency of information transmission, which provides a coherent combing gain, or also called as the beam-forming gain. This gain is quantitively reflected from the larger singular value of the DTM of the common source as demonstrated in example~\ref{example:MAC}. Finally, the overall transmitted codeword is the superposition of all such sources. 
\end{remark}

\section{Conclusion} \label{sec:con}

In this paper, we developed a local approximation approach, which approximates the K-L divergence by a squared Euclidean metric in an Euclidean space. Under the local approximation, we constructed the coordinates and the notion of orthogonality for the probability distribution space. With this approach, we can solve a certain class of information theory problems, which we call the linear information coupling problems, for different types communication channels. We also showed that the single-letterization of multi-terminal problems can be simplified as some linear algebra problems, which can be solved in a systematic way. Moreover, applying our approach to the general broadcast channels, the transmission of the common message can be formulated as the trade-off between multiple linear systems. In order to achieve the optimal trade-off, it is required to code over multiple letters with the number of letters proportional to the number of receivers. Finally, for the multiple access channel with common sources, there exists some coherent combing gains due to the cooperation between transmitters, and we can evaluate this gain quantitively by using our technique.

The development of this paper can be extended to more general multi-terminal communication networks, such as layered networks. Thus, our approach can be considered as a useful tool to study efficient information flow in networks. Moreover, the coding insights obtained in this paper can be easily carried to designing efficient network communication channel codes. This provides interesting directions for future researches.


\appendices

\section{Proof of Lemma~\ref{prop:BC}} \label{app:lem3}


Since $U$, $V_1$, and $V_2$ are mutually independent, we have
\begin{align*}
\notag
&I( U, V_1 , V_2; \uX ) \\
&= I(U ; \uX) + I(V_1 ; \uX | U) + I(V_2 ; \uX | U, V_1) \\ 
&\geq I(U ; \uX) + I(V_1 ; \uX) + I(V_2 ; \uX).
\end{align*}
Thus, the rate region \eqref{LICP-BC1} is belong to \eqref{LICP-BC2}. On the other hand, for any rate tuple $(R_0, R_1, R_2)$ in (\ref{LICP-BC2}) achieved by some mutually independent $U$, $V_1$, and $V_2$, with $\frac{1}{n} I(U;\uX) \leq \frac{1}{2} \epsilon_0^2$, and $\frac{1}{n} I(V_i;\uX) \leq \frac{1}{2} \epsilon_i^2$ for $i = 1,2$, where $\sum_{i=1}^{3} \epsilon_i^2 = \epsilon^2$, we assume that the conditional distributions achieving this rate tuple have the perturbation forms $P_{\uX | U = u} = P_{\uX} + \sqrt{n} \epsilon_0 \cdot J_{u}$, and $P_{\uX | V_i = v_i} = P_{\uX} + \sqrt{n} \epsilon_i \cdot J_{v_i}$, for $i = 1,2$. Then, it is easy to verify that
\begin{align*}
Q_{\uX | ( U , V_1 , V_2 ) = ( u , v_1 , v_2 ) }= P_{\uX} + \sqrt{n} \epsilon_0 J_{u} + \sqrt{n} \epsilon_1 J_{v_1} + \sqrt{n} \epsilon_2 J_{v_2}
\end{align*}
is a valid conditional distribution with marginal $P_{\uX}$. Therefore, using $Q_{\uX | ( U , V_1 , V_2 ) = ( u , v_1 , v_2 ) }$ as the conditional distribution, the mutual information
\begin{align} \notag
&\frac{1}{n} I( U , V_1 , V_2 ; \uX ) \\ \notag
&= \frac{1}{2} \sum_{u,v_1,v_2} P_U(u) P_{V_1}(v_1) P_{V_2}(v_2)  \cdot \left\| \epsilon_0 J_{u} + \epsilon_1 J_{v_1} + \epsilon_2 J_{v_2} \right\|_{P_X}^2 + o(\epsilon^2) \\ \label{eq:123}
&\leq \frac{1}{2} \left( \epsilon_0^2 +  \epsilon_1^2 +  \epsilon_2^2 \right) + o(\epsilon^2) = \frac{1}{2} \epsilon^2 + o(\epsilon^2),
\end{align}
where \eqref{eq:123} is resulted from the definition of the perturbation vectors:
\begin{align*}
\sum_u P_U(u)  J_{u} = 0, \quad \sum_{v_i} P_{V_i} (v_i)  J_{v_i}  = 0, \ \mbox{for} \ i = 1,2, 
\end{align*}
and
\begin{align*}
&\sum_u P_U(u) \left\|  J_{u}  \right\|_{P_X}^2 \leq 1,\\  
&\sum_{v_i} P_{V_i} (v_i) \left\|  J_{v_i}  \right\|_{P_X}^2 \leq 1, \ \mbox{for} \ i = 1,2. 
\end{align*}
Hence, we can take $Q_{\uX | ( U , V_1 , V_2 ) = ( u , v_1 , v_2 ) }$ as the conditional distribution in \eqref{LICP-BC1}, and obtain a rate tuple that is equal to $(R_0, R_1, R_2)$, up to the first order approximation.

\section{Proof of Theorem \ref{thm:BC}} \label{ap:proof 1}



The part a) and b) of Theorem~\ref{thm:BC} are directly followed from the following lemma.
\begin{lemma} [Optimal solutions have product forms] \label{lem:pf}
For any $n\geq K$, there exist $| \cX |$-dimensional vectors $\uv^*_1, \uv^*_2, \ldots , \uv^*_K$, which satisfy $\sum_{u = 1}^K \| \uv^*_u \|^2 = 1$, and $\langle \uv_0 , \uv^*_i \rangle = 0$, for $u = 1,2, \ldots , K$, such that 
\begin{align} \label{lemma}
\K_u = \uv_0^{(n-K)} \otimes  \left( \uv_0^{(u-1)} \otimes \uv_u^* \otimes \uv_0^{(K - u)} \right),
\end{align} 
and $|\cU| = K$ with $P_U(u) = 1/K$ is an optimal solution of~\eqref{eq:k-BC}. Here, $\uv_0^{(\cdot)}$ is the Kronecker product of $\uv_0$.  
\end{lemma}

In order to prove lemma~\ref{lem:pf}, we need to first establish the following lemma~\ref{lem:dim_count}, which illustrates the required degree of freedom for describing the optimal tradeoff between multiple linear systems.

\begin{lemma} \label{lem:dim_count}
Assume that $\Theta_i = \diag \left\{ \theta_{i,1} , \theta_{i,2} , \ldots , \theta_{i,M} \right\}$, for $i = 1,2, \ldots , k$, if the optimization problem
\begin{align} \label{eq:lemma_6}
\max. \quad &\| \Theta_k \cdot \uc \|^2 \\ \notag
\ {\rm subject \ to}: \quad &\| \uc \|^2 = 1, \  {\rm and} \  \| \Theta_i \cdot \uc \|^2 = \lambda_i, \\ \notag
&{\rm for} \ i = 1,2,\ldots, k-1,
\end{align}
has global optimal solutions, then there exists a global optimal solution $\uc^* = \left[ c^*_1 \ c^*_2 \ \ldots \ c^*_M \right]^T$ with at most $k$ nonzero entries.
\end{lemma}

\begin{proof}  [Proof of Lemma \ref{lem:dim_count}]
Let us assume that $\uc^* = \left[ c^*_1 \ c^*_2 \ \ldots \ c^*_M \right]^T$ is the global optimal solution of \eqref{eq:lemma_6} with the least number of nonzero entries $l$. If $l > k$, then without loss of generality, we can assume that $c_i \neq 0$ for $i \leq l$, and $c_i = 0$ for $i > l$. For $i = 0, 1, \ldots , k$, define the vector $\uth_{i,l} = \left[ \theta_{i,1}^2 \ \theta_{i,2}^2 \ \ldots \ \theta_{i,l}^2 \right]^T \in \mathbb{R}^l$, and $\uth_{0,l} = \left[ 1 \ \ldots \ 1 \right]^T \in \mathbb{R}^l$, then the null space $\mbox{Null} \left( \uth_{i,l} \right) \subset \mathbb{R}^l$ has dimension at least $l-1$, and the space of the intersecting of null spaces $\cup_{i = 0}^{k-1} \mbox{Null} \left( \uth_{i,l} \right)$ has dimension at least $l - k >0$. Let the nonzero vector $\ud = \left[ d_1 \ d_2 \ \ldots \ d_l \right]^T \in \cup_{i = 0}^{k-1} \mbox{Null} \left( \uth_{i,l} \right)$, and $\uc_t = \left[ c_{1,t} \ \ldots \ c_{l,t} \ 0 \ \ldots \ 0 \right]^T \in \mathbb{R}^m$, where $c_{i,t} = \sqrt{c_i^{*2} + t \cdot d_i}$. Since $\uc^*$ is a global optimal, we have ${\partial \| \Theta_k \cdot \uc_t \|^2 \over \partial t} = 0$, which implies that $\langle \uth_{k,l} , \ud \rangle = 0$. Thus, if we take $t^* \triangleq \max \left\{ t :  c_i^{*2} + t \cdot d_i \geq 0, \ \forall 1 \leq i \leq l  \right\}$, then $\uc_{t^*}$ is a global optimal solution with at most $l-1$ nonzero entries, which contradicts to the assumption of $\uc^*$. Therefore, $\uc^*$ has at most $k$ nonzero entries.
\end{proof}


Lemma~\ref{lem:dim_count} immediately implies the following Corollary.

\begin{thmcorol} \label{cor:1}
Assume that $\Theta_i = \diag \left\{ \theta_{i,1} , \theta_{i,2} , \ldots , \theta_{i,M} \right\}$, for $i = 1,2, \ldots , k$, then the optimization problem
\begin{align}
\max_{ \| \uc \|^2 = 1} \min_{1 \leq i \leq k} \left\{ \| \Theta_i \cdot \uc \|^2  \right\}
\end{align}
has a global optimal solution $\uc^* = \left[ c^*_1 \ c^*_2 \ \ldots \ c^*_M \right]^T$ with at most $k$ nonzero entries.
\end{thmcorol}

Now, let us prove Lemma \ref{lem:pf}. For the sake of presentation convenience, we demonstrate here the proof of the case $K=2$, while this technique can be carried out to general $K$-user case without difficulty.

\begin{proof}  [Proof of Lemma \ref{lem:pf}]
Let $\sigma_0 =1 , \sigma_1 , \ldots , \sigma_m$ and $\mu_0 = 1 , \mu_1 , \ldots , \mu_m$ be the singular values of the DTM's $B_1$ and $B_2$, respectively, and the corresponding right singular vectors are $\uv_0 ,\uv_1,...,\uv_m$ and $\uu_0 ,\uu_1,...,\uu_m$, where $m \triangleq | \cX | - 1$, and $\uv_0 = \uu_0 = \left[ \sqrt{P_X}, x \in \cX \right]^T$. Moreover, we assume that $\uv_i = \sum_{k = 0}^m \phi_{ik} \uu_k$, where $\Phi_0 = [\phi_{ij}]_{i,j=0}^m$ is an $|\cX|$-by-$|\cX|$ unitary matrix. Then, we have $\phi_{00} = 1$, $\phi_{i0} = \phi_{0i} = 0$, for all $i>0$. 

Now for any $n \geq K$, our goal here is to show that for any choice of $\K'_u$ satisfying the constraints of~\eqref{eq:k-BC}, we can always find a set of $\K_u$ taking product form~\eqref{lemma} that has no smaller output image in~\eqref{eq:k-BC}. The first observation here is that we can only restrict our focus on the $\K'_u$'s such that all the $\K'_u$'s are mutually orthogonal. The reason is that we can construct mutually orthogonal vectors $\K''_u = \uv_0^{((u-1) \cdot n)} \otimes \K'_u \otimes \uv_0^{((|\cU|-u)\cdot n)}$, and alternatively prove lemma~\ref{lem:pf} for $\K''_u$. Now, for mutually orthogonal vectors $\K'_u$, let us write $\K'$ as
\begin{align*}
\K' = \sum_u  \sqrt{P_U(u)} \cdot  \K_{u}.
\end{align*}
Then, the output images of~\eqref{eq:k-BC} becomes $\| B_i^{(n)} \cdot \K' \|$, for $i = 1,2$, with the constraint $\| \K' \| = 1$, and $\langle \K', \uv_0^{(n)}  \rangle =0$.

Now, suppose that 
\begin{align*}
\K' &= \mathop{\sum_{i_1, \ldots , i_n=0}^m}_{(i_1, \ldots , i_n)\neq (0, \ldots ,0)} \alpha_{i_1 \cdots i_n} \cdot \left( \uv_{i_1} \otimes \cdots \otimes \uv_{i_n} \right)\\ 
&= \sum_{j_1, \ldots , j_n=0}^m \beta_{j_1 \cdots j_n} \cdot \left( \uu_{j_1} \otimes \cdots \otimes \uu_{j_n} \right),
\end{align*}
where
\begin{align*}
\beta_{j_1 \cdots j_n} = \mathop{\sum_{i_1, \ldots , i_n=0}^m}_{(i_1, \ldots , i_n)\neq (0, \ldots ,0)} \alpha_{i_1 \cdots i_n} \phi_{i_1j_1} \cdots \phi_{i_n j_n}.
\end{align*}
Then, since $\| \K' \| = 1$, we have
\begin{align*}
\mathop{\sum_{i_1, \ldots , i_n=0}^m}_{(i_1, \ldots , i_n)\neq (0, \ldots ,0)} \| \alpha_{i_1 \cdots i_n} \|^2 = \sum_{j_1, \ldots , j_n=0}^m \| \beta_{j_1 \cdots j_n} \|^2 = 1.
\end{align*}
Now, let us define 
\begin{align*}
\cS_k = \{ &\left( i_1 , \ldots , i_n \right) : 0 \leq i_a \leq m, \ \forall \ 1 \leq a \leq n,\\ 
&i_k \neq 0, \ i_{k+1} = \cdots = i_n = 0 \},
\end{align*}
\begin{align*}
\cT_k =  \left\{ \left( i_1 , \ldots , i_k \right) : 0 \leq i_a \leq m, \ \forall \ 1 \leq a \leq k \right\},
\end{align*}
with $\cT_0 = \emptyset , \ \cT = \cup_{k = 0}^{n-1} \cT_k,$ and 
$$\xi : \cT \mapsto \left\{ 1, 2 , \ldots , M \right\}$$ 
is a bijective map, where $M = |\cT|$. Then,
\begin{align*} 
\notag
\left\| B_1^{(n)} \cdot \K' \right\|^2
&= \mathop{\sum_{i_1, \ldots , i_n=0}^m}_{(i_1, \ldots , i_n)\neq (0, \ldots ,0)} \alpha_{i_1 \cdots i_n}^2 \sigma_{i_1}^2 \cdots \sigma_{i_n}^2 
\leq \sum_{k = 1}^n \sum_{\left( i_1 , \ldots , i_n \right) \in \cS_k } \alpha_{i_1 \cdots i_n}^2 \sigma_{i_k}^2 \\ \notag 
&= \sum_{k = 1}^n \sum_{\left( i_1 , \ldots , i_{k-1} \right) \in \cT_{k-1} } \sum_{i_k = 1}^m \alpha_{i_1 \cdots i_{k-1} i_k 0  \ldots  0}^2 \sigma_{i_k}^2 \\
&= \sum_{k = 1}^n \sum_{\left( i_1 , \ldots , i_{k-1} \right) \in \cT_{k-1} } \| \Sigma \cdot \ual_{\xi \left( i_1 , \ldots , i_{k-1} \right)} \|^2 \\ \notag
&= \sum_{i = 1}^{M} \| \Sigma \cdot \ual_{i} \|^2,
\end{align*}
where $\Sigma = \mbox{diag} \left\{\sigma_1, \ldots , \sigma_m  \right\}$, and $\ual_{\xi \left( i_1 , \ldots , i_{k-1} \right)}$ is defined as
\begin{align*}
\ual_{\xi \left( i_1 , \ldots , i_{k-1} \right)}
= \left[ \alpha_{i_1 \cdots i_{k-1} 1 0 \cdots 0} \ \alpha_{i_1 \cdots i_{k-1} 2 0 \cdots 0} \cdots \alpha_{i_1 \cdots i_{k-1} m 0 \cdots 0} \right]^T.
\end{align*}
Moreover,
\begin{align}
\notag
&\left\| B_2^{(n)} \cdot \K' \right\|^2 
= \sum_{j_1, \ldots , j_n=0}^m \beta_{j_1 \cdots j_n}^2 \mu_{j_1}^2 \cdots \mu_{j_n}^2 \\ \notag
&= \sum_{j_1, \ldots , j_n=0}^m \mu_{j_1}^2 \cdots \mu_{j_n}^2 \cdot  \left( \sum_{k = 1}^n \sum_{\left( i_1 , \ldots , i_n \right) \in \cS_k } \alpha_{i_1 \cdots i_n} \phi_{i_1j_1} \cdots \phi_{i_n j_n}  \right)^2 \\ \label{eq:crazy}
&= \sum_{k = 1}^n \sum_{j_1, \ldots , j_n=0}^m \mu_{j_1}^2 \cdots \mu_{j_n}^2  \cdot  \left(  \sum_{\left( i_1 , \ldots , i_n \right) \in \cS_k } \alpha_{i_1 \cdots i_n} \phi_{i_1j_1} \cdots \phi_{i_n j_n}  \right)^2 ,
\end{align}
where \eqref{eq:crazy} is because for any $1 \leq k_1 < k_2 \leq n$,
\begin{align*}
&\sum_{j_1, \ldots , j_n=0}^m \left(  \sum_{\left( i_1 , \ldots , i_n \right) \in \cS_{k_1} } \alpha_{i_1 \cdots i_n} \phi_{i_1j_1} \cdots \phi_{i_n j_n}  \right)  \cdot \left(  \sum_{\left( i_1 , \ldots , i_n \right) \in \cS_{k_2} } \alpha_{i_1 \cdots i_n} \phi_{i_1j_1} \cdots \phi_{i_n j_n}  \right) \mu_{j_1}^2 \cdots \mu_{j_n}^2 \\ 
=& \sum_{j_1, \ldots , j_n=0}^m \mathop{\sum_{\left( i_1 , \ldots , i_{k_1} , 0 \ldots , 0 \right) \in \cS_{k_1}}}_{\left( i_1' , \ldots , i_{k_2}' , 0 \ldots , 0 \right) \in \cS_{k_2}}  
\alpha_{ i_1 \cdots  i_{k_1}  0 \cdots 0} \alpha_{ i_1' \cdots i_{k_2}' 0 \cdots  0} \cdot \left( \phi_{i_1j_1}  \phi_{i_1'j_1} \right) \cdots \underbrace{ \left( \phi_{0 j_{k_2}}  \phi_{i_{k_2}'j_{k_2}} \right) }_{\mathclap{\mbox{nonzero only if $j_{k_2} = i_{k_2}' = 0$, but  $i_{k_2}' \neq 0$}  } }  \mu_{j_1}^2 \cdots \mu_{j_n}^2\\
=& \ 0
\end{align*}
Then, let $\Omega = \mbox{diag} \left\{\mu_1, \ldots , \mu_m  \right\}$ and $\Phi = [\phi_{ij}]_{i,j = 1}^m$, for each $k$, we have
\begin{align} \notag
& \sum_{j_1, \ldots , j_n=0}^m \left(  \sum_{\left( i_1 , \ldots , i_n \right) \in \cS_k } \alpha_{i_1 \cdots i_n} \phi_{i_1j_1} \cdots \phi_{i_n j_n}  \right)^2  \mu_{j_1}^2 \cdots \mu_{j_n}^2 \\ \notag
&\leq  \sum_{j_1, \ldots , j_n=0}^m \left(  \sum_{\left( i_1 , \ldots , i_n \right) \in \cS_k } \alpha_{i_1 \cdots i_n} \phi_{i_1j_1} \cdots \phi_{i_n j_n}  \right)^2 \mu_{j_k}^2 \\ \notag
&= \sum_{j_1, \ldots , j_n=0}^m \sum_{\left( i_1 , \ldots , i_n \right) \in \cS_k} \sum_{\left( i_1' , \ldots , i_n' \right) \in \cS_k} \alpha_{i_1 \cdots i_n} \alpha_{i_1' \cdots i_n'}  \cdot \left( \phi_{i_1 j_1} \phi_{i_1' j_1} \right) \cdots \left( \phi_{i_n j_n} \phi_{i_n' j_n} \right)  \mu_{j_k}^2 \\  \label{eq:crazy2}
&= \sum_{j_k =0}^m \sum_{\left( i_1 , \ldots , i_{k-1} \right) \in \cT_{k-1}} \left( \sum_{i_k=1}^m \alpha_{i_1\cdots i_k0 \cdots 0} \phi_{i_k j_k} \right)^2 \mu_{jk}^2 \\ \notag
&= \sum_{\left( i_1 , \ldots , i_{k-1} \right) \in \cT_{k-1}} \| \Omega \Phi^T \ual_{\xi \left( i_1 , \ldots , i_{k-1} \right)} \|^2,
\end{align}
where \eqref{eq:crazy2} is because $\Phi$ is a unitary matrix, and
\begin{align*}
\sum_{j_r = 0}^m \phi_{i_r j_r} \phi_{i_r' j_r} = 
\left\{
\begin{array}{ll}
1 & \mbox{if} \ i_r = i'_r \\
0 & \mbox{otherwise}
\end{array} \right. .
\end{align*}
Therefore, \eqref{eq:crazy} becomes
\begin{align*}
\left\| B_2^{(n)} \cdot \K' \right\|^2 &\leq \sum_{k = 1}^n \sum_{\left( i_1 , \ldots , i_{k-1} \right) \in \cT_{k-1}} \| \Omega \Phi^T \ual_{\xi \left( i_1 , \ldots , i_{k-1} \right)} \|^2 = \sum_{i = 1}^{M} \| \Omega \Phi^T \ual_i \|^2.
\end{align*}

Now, let us define $\Theta_1 = \diag \left\{ \theta_{1,1} ,  \ldots , \theta_{1,M} \right\}$, and $\Theta_2 = \diag \left\{ \theta_{2,1} , \ldots , \theta_{2,M} \right\}$ as 
\begin{align*}
&\theta_{1,i} = 
\left\{ \begin{array}{cl} \frac{ \| \Sigma  \ual_i \|}{\| \ual_i \|}  &\mbox{ if $\| \ual_i \| \neq 0$} \\
  0 &\mbox{ otherwise}
       \end{array} \right. , 
\\
&\theta_{2,i} = 
\left\{ \begin{array}{cl} \frac{  \| \Omega \Phi^T  \ual_i \| }{\| \Phi^T \cdot  \ual_i \|}  &\mbox{ if $\| \Phi^T \cdot  \ual_i \| \neq 0$} \\
  0 &\mbox{ otherwise}
       \end{array} \right. .
\end{align*}
Then, from Corollary \ref{cor:1}, there exists an optimal solution $\uc^*$ of the optimization problem
\begin{align} \label{eq:1}
\max_{\uc \in \mathbb{R}^{M} : \| \uc \|^2 = 1} \min \left\{  \| \Theta_1 \cdot \uc \|^2 ,  \| \Theta_2 \cdot \uc \|^2  \right\},
\end{align}
with at most two nonzero entries. Let the $i_1$-th entry $c^*_{i_1}$ and the $i_2$-th entry $c^*_{i_2}$ of $\uc^*$ are nonzero. Note that $\|\ual_i\| = \|\Phi^T \cdot \ual_i\|$, for all $1 \leq i \leq M$, and
\begin{align*}
\sum_{i = 1}^{M} \|\ual_i\|^2 = \mathop{\sum_{i_1, \ldots , i_n=0}^m}_{(i_1, \ldots , i_n)\neq (0, \ldots ,0)} \| \alpha_{i_1 \cdots i_n} \|^2 = 1, 
\end{align*}
thus the vector $\ual = \left[ \|\ual_1\| \ \|\ual_2\| \ \ldots \ \|\ual_{M}\| \right]^T$ has unit norm. This implies that
\begin{align*}
&\min \left\{ \| \Theta_1 \cdot \uc^* \|^2 ,  \| \Theta_2 \cdot \uc^* \|^2 \right\} \\
\geq& \min \left\{ \| \Theta_1 \cdot \ual \|^2 ,  \| \Theta_2 \cdot \ual \|^2 \right\} \\
=& \min \left\{ \sum_{i = 1}^{M}  \| \Sigma  \ual_i \|^2 ,  \sum_{i = 1}^{M}  \| \Omega \Phi^T  \ual_i \|^2 \right\} \\
=& \min \left\{ \| B_1^{(n)} \cdot L \|^2 , \| B_2^{(n)} \cdot L \|^2 \right\}.
\end{align*}
Now, let us take vectors $\uv^*_1 = c^*_{i_1} \cdot \sum_{j = 1}^m  \frac{\ual_{i_1}(j)}{\| \ual_{i_1} \|}   \uv_j $, and $\uv^*_2 = c^*_{i_2} \cdot \sum_{j = 1}^m  \frac{\ual_{i_2}(j)}{\| \ual_{i_2} \|}   \uv_j $, where $\ual_{i_1}(j)$ and $\ual_{i_2}(j)$ are the $j$-th entries of $\ual_{i_1}$ and $\ual_{i_2}$, respectively. Then, the vector $\K^* = \uv_0^{(n-2)} \otimes \left( \uv^*_1 \otimes \uv_0 + \uv_0 \otimes \uv^*_2 \right)$ satisfies $\| \K^* \| = 1$, and
\begin{align*}
&\min \left\{ \| B_1^{(n)} \cdot \K^* \|^2 , \| B_2^{(n)} \cdot \K^* \|^2 \right\} \\ 
=& \min \left\{ \| \Theta_1 \cdot \uc^* \|^2 ,  \| \Theta_2 \cdot \uc^* \|^2 \right\} \\
\geq& \min \left\{ \| B_1^{(n)} \cdot \K' \|^2 , \| B_2^{(n)} \cdot \K' \|^2 \right\}.
\end{align*}
Therefore, by taking $\K_i = \uv_0^{(n-2)} \otimes \uv^*_i \otimes \uv_0$ for $i=1,2$ in~\eqref{lemma}, we come up with vectors with product form and no smaller output images. This proves lemma~\ref{lem:pf}.
\end{proof}

Now, in order to prove part c), we only need to show that for $K=2$, $\sup_n \lambda_B^{(n)} = \lambda_B^{(1)}$. Equivalently, we will show that there exists a unit vector $\K$ such that $\sup_n \lambda_B^{(n)} = \| B^{(1)} \cdot \K \|$. For this purpose, let us start from an optimal solution $\K^* = \uv_1 \otimes \uv_0 + \uv_0 \otimes \uv_2$ of~\eqref{eq:k-BC}, where $\uv_1 = \sum_{j = 1}^m a_{j} \cdot \uv_j$, and $\uv_2 = \sum_{j = 1}^m b_{j} \cdot \uv_j$, and $\| \uv_1 \|^2 + \| \uv_1 \|^2 = 1$. Let $\ua = \left[ a_{1}  \ \ldots \ a_{m} \right]^T$ and $\ub = \left[ b_{1}  \ \ldots \ b_{m} \right]^T$, then,
\begin{align*}
&\min \left\{ \|  B_1^{(2)} \cdot L_2 \|^2 , \|  B_2^{(2)} \cdot L_2 \|^2 \right\} \\
&= \min \left\{ \| \Sigma \ua \|^2 + \| \Sigma \ub \|^2 , \| \Omega \Phi^T \ua \|^2 + \| \Omega \Phi^T \ub \|^2 \right\}
\end{align*}
Our goal is to show that there exists a unit vector $\uc = \left[ c_{1} \ \ldots \ c_{m} \right]^T$ such that
\begin{align} \label{eq:gamma}
\left\{
\begin{array}{ll}
\| \Sigma \uc \|^2 \geq \| \Sigma \ua \|^2 + \| \Sigma \ub \|^2, \\ 
\| \Omega \Phi^T \uc \|^2 \geq \| \Omega \Phi^T \ua \|^2 + \| \Omega \Phi^T \ub \|^2.
\end{array}
\right.
\end{align}
Then, taking $\K = \sum_{j = 1}^m c_{j} \cdot \uv_j$, and we are done.

Now for $m = 2$, we want to consider the vectors $\uc_1 = [\sqrt{a_1^2 + b_1^2} \ \sqrt{a_2^2 + b_2^2}]^T$, and $\uc_2 = [\sqrt{a_1^2 + b_1^2} \ -\sqrt{a_2^2 + b_2^2}]^T$. Obviously, $\| \Sigma \uc_1 \|^2 = \| \Sigma \uc_2 \|^2 = \| \Sigma \ua \|^2 + \| \Sigma \ub \|^2$. Note that
\begin{align*}
& \left( \| \Omega \Phi^T \ua \|^2 + \| \Omega \Phi^T \ub \|^2 - \| \Omega \Phi^T \uc_1 \|^2  \right) \cdot \left( \| \Omega \Phi^T \ua \|^2 + \| \Omega \Phi^T \ub \|^2 - \| \Omega \Phi^T \uc_2 \|^2 \right) \\
= & -4 \left( a_1 b_2 - b_1 a_2 \right)^2 \cdot \left(\phi_{11}\phi_{21}\mu_1^2 + \phi_{12}\phi_{22}\mu_2^2 \right)^2 \leq 0
\end{align*}
Therefore, at least one of $\uc_1$ and $\uc_2$ satisfies \eqref{eq:gamma}. 

On the other hand, for $m > 2$, let $A = \left[ \ua \ \ub \right]$, and consider the SVD of matrices $\Sigma A = U_1^T \Sigma_1 V_1$, and $\Omega \Phi^T A = U_2^T \Sigma_2 V_2$, where $V_i$ is a $2 \times 2$ unitary matrix, and $\Sigma_i$ is a diagonal matrix, for $i = 1,2$. Moreover, we denote $V_1 = [\varphi_1 \ \varphi_2]$, and $V_2 = [\eta_1 \ \eta_2]$, where $\varphi_i$ and $\eta_i$ are all two dimensional unit vectors. Then, $\| \Sigma \ua \|^2 =  \| \Sigma_1 \varphi_1 \|^2$, $\| \Sigma \ub \|^2 =  \| \Sigma_1 \varphi_2 \|^2$, and $\| \Omega \Phi^T \ua \|^2 = \| \Sigma_2 \left( V_2 V_1^{-1} \right) \eta_1 \|^2$, $\| \Omega \Phi^T \ub \|^2 = \| \Sigma_2 \left( V_2 V_1^{-1} \right) \eta_2 \|^2$. Since $V_2 V_1^{-1}$ is a unitary matrix, there exists a two dimensional unit vector $\uc'$, such that
\begin{align*}
\left\{
\begin{array}{ll}
\| \Sigma_1 \uc' \|^2 \geq \| \Sigma_1 \varphi_1 \|^2 + \| \Sigma_1 \varphi_2 \|^2, \\
\| \Sigma_2 \left( V_2 V_1^{-1} \right) \uc' \|^2  \geq \| \Sigma_2 \left( V_2 V_1^{-1} \right) \eta_1 \|^2 + \| \Sigma_2 \left( V_2 V_1^{-1} \right) \eta_2 \|^2 .
\end{array}
\right.
\end{align*}
Now, taking $\uc = A V_1^{-1} \uc'$, then $\| \Sigma \uc \|^2 = \| \Sigma_1 \uc' \|^2$, and $\| \Omega \Phi^T \uc \|^2 = \| \Sigma_2 \left( V_2 V_1^{-1} \right) \uc' \|^2$, which implies that $\uc$ satisfies (\ref{eq:gamma}). 

\section{Proof of Theorem \ref{thm:MAC}} \label{ap:proof 2}



In order to prove Theorem~\ref{thm:MAC}, we will show that: 
\begin{itemize}
\item [1)] the second largest singular value of $B_{0,n} = \left[ B_1^{(n)} \ \ldots \ B_k^{(n)} \right]$ is the same as $B_0$
\item [2)] $\K_i$ is orthogonal to $\left[ \sqrt{P_{X_i}}, x_i \in \cX_i \right]^T$, for all $1 \leq i \leq k$.
\end{itemize}

First, let us prove 1). Suppose that $\sigma_{1,n}$ is the second largest singular value of $B_{0,n}$, then we want to show that $\sigma_{1,n} = \sigma_{1,1}$, for all $n>1$. Let us first show that $\sigma_{1,2} = \sigma_{1,1}$. Observe that $\left[ \sqrt{P_{Y}}, y \in \cY \right]^T$ is the left singular vector of $B_0$, corresponding to the largest singular value $\sqrt{k}$, thus we can assume that the singular values of $B_0$ are $\sigma_{0,1} = \sqrt{k} \geq \sigma_{1,1} \geq \sigma_{2,1} \geq \cdots \geq \sigma_{m,1}$, and the corresponding left singular vectors are $\uw_0 = \left[ \sqrt{P_{Y}}, y \in \cY \right]^T , \uw_1 , \uw_2 , \ldots , \uw_m$, where $m \triangleq \min \left\{ \sum_{i = 1}^k |\cX_i| , |\cY| \right\} -1$. Note that for all $0 \leq i \leq m$,
\begin{align*}
& \left( \uw_0 \otimes \uw_i \right)^T \cdot \left( B_{0,2} B_{0,2}^T \right) \\
=&  \left( \uw_0^T \otimes \uw_i^T \right) \cdot \left( B_1 B_1^T  \otimes B_1 B_1^T  + B_2 B_2^T \otimes B_2 B_2^T \right) \\ 
=&  \uw_0^T \otimes \left[ \uw_i^T \cdot \left( B_1 B_1^T + B_2 B_2^T \right) \right] \\
=& \uw_0^T \otimes \left[ \uw_i^T \cdot \left( B_0 B_0^T \right) \right] = \sigma_{i,1}^2 \cdot \uw_0^T \otimes \uw_i^T,
\end{align*}
therefore, $\uw_0 \otimes \uw_i$ is a singular vector of $B_{0,2}$, with singular value $\sigma_{i,1}$. Similarly, $\uw_j \otimes \uw_0$ is a singular vector of $B_{0,2}$, with singular value $\sigma_{j,1}$, for all $0 \leq j \leq m$. Hence, in order to show that $\sigma_{1,2} = \sigma_{1,1}$, we only need to show that for any unit vector $\uw \in \mbox{span} \left\{ \uw_i \otimes \uw_j , \ 1 \leq i,j \leq m \right\}$, $\| \uw^T \cdot B_{0,2} \| \leq \sigma_{1,2}$. To this end, note that $\| \uw^T \cdot \left( B_0 \otimes B_0 \right) \| \leq \sigma_{1,2}^2$, therefore
\begin{align*}
\| \uw^T \cdot B_{0,2} \|^2 
& \leq \| \uw^T \cdot B_{0,2} \|^2 + \| \uw^T \cdot \left( B_1 \otimes B_2 \right) \|^2 + \| \uw^T \cdot \left( B_2 \otimes B_1 \right) \|^2 \\ 
& = \| \uw^T \cdot \left( B_0 \otimes B_0 \right) \|^2 \leq \sigma_{1,2}^4 \leq \sigma_{1,2}^2.
\end{align*}
Thus, we have $\sigma_{1,2} = \sigma_{1,1}$. With the same arguments, we can show that for any positive integer $N$, $\sigma_{1,2^N} = \sigma_{1,1}$. Since $\sigma_{1,n}$ is non-decreasing with $n$, this implies that $\sigma_{1,n} = \sigma_{1,1}$, for all $n$.

Now, let us prove the statement (ii). For simplicity, we denote $\uv_{i,0} = \left[ \sqrt{P_{X_i}}, x_i \in \cX_i \right]^T$, and $\uv = \left[ \uv_{1,0}^T \ \ldots \ \uv_{k,0}^T \right]^T$, then $\uv$ is the singular vector of $B_0$, corresponding to the largest singular value $\sqrt{k}$. Suppose that $\langle \K_i , \uv_{i,0} \rangle = \mathbb{I}_i$, since $\K_{U=u}$ is orthogonal to $\uv$, we have $\sum_{i = 1}^k \mathbb{I}_i = 0$. Now, if there exits a $j$ such that $\mathbb{I}_j \neq 0$, then define the vector $\tilde{\K}_{U=u} = \left[ \tilde{\K}_1^T \ \ldots \ \tilde{\K}_k^T \right]^T$, where $\tilde{\K}_i = \left( \K_i - \mathbb{I}_i \cdot \uv_i \right)/\sqrt{1-\mathbb{I}}$, and $\mathbb{I} = \sum_{i = 1}^k \mathbb{I}_i^2 >0$. This definition of $\tilde{\K}_u$ is valid because $\mathbb{I} < \sum_{i = 1}^k \| \K_i \|^2 = 1$. Then, it is easy to verify that $\| \tilde{\K}_{U=u} \| = 1$, and $\tilde{\K}_{U=u}$ is orthogonal to $\uv$. Moreover, 
\begin{align*}
B_0 \cdot \tilde{\K}_{U=u} & = \sqrt{1-\mathbb{I}}^{-1} \cdot \left( B_0 \cdot \K_{U=u} - \sum_{i = 1}^k \mathbb{I}_i \cdot \left( B_i \cdot \uv_{i,0} \right) \right) \\
& =  \sqrt{1-\mathbb{I}}^{-1} \cdot \left( B_0 \cdot \K_{U=u} - \left( \sum_{i = 1}^k \mathbb{I}_i \right) \cdot \uw_0 \right) \\ \notag 
&=  \sqrt{1-\mathbb{I}}^{-1} \cdot \left( B_0 \cdot \K_{U=u} \right),
\end{align*}
where $\uw_0 = \left[ \sqrt{P_{Y}}, y \in \cY \right]^T$. Therefore, $\| B_0 \cdot \tilde{\K}_{U=u} \| > \| B_0 \cdot \K_{U=u} \|$, since $\mathbb{I} > 0$. This contradicts to the assumption that $\K_{U=u}$ is the singular vector of $B_0$, corresponding to the second largest singular value. Thus, $\K_i$ is orthogonal to $\left[ \sqrt{P_{X_i}}, x_i \in \cX_i \right]^T$, for all $1 \leq i \leq k$.



\end{document}

%% file: example1_3.pdf_tex

\begingroup
  \makeatletter
  \providecommand\color[2][]{%
    \errmessage{(Inkscape) Color is used for the text in Inkscape, but the package 'color.sty' is not loaded}
    \renewcommand\color[2][]{}%
  }
  \providecommand\transparent[1]{%
    \errmessage{(Inkscape) Transparency is used (non-zero) for the text in Inkscape, but the package 'transparent.sty' is not loaded}
    \renewcommand\transparent[1]{}%
  }
  \providecommand\rotatebox[2]{#2}
  \ifx\svgwidth\undefined
    \setlength{\unitlength}{181.41732178pt}
  \else
    \setlength{\unitlength}{\svgwidth}
  \fi
  \global\let\svgwidth\undefined
  \makeatother
  \begin{picture}(1,0.68749996)%
    \put(0,0){\includegraphics[width=\unitlength]{example1_3.pdf}}%
    \put(0.09247512,0.64444346){\color[rgb]{0,0,0}\makebox(0,0)[lb]{\smash{\Large$X$}}}%
    \put(0.77539009,0.64482162){\color[rgb]{0,0,0}\makebox(0,0)[lb]{\smash{\Large$Y$}}}%
    \put(0.08416668,0.5574618){\color[rgb]{0,0,0}\makebox(0,0)[lb]{\smash{\blue\Large$1$}}}%
    \put(0.08375533,0.24638367){\color[rgb]{0,0,0}\makebox(0,0)[lb]{\smash{\red\large$2$}}}%
    \put(0.08219626,0.02933889){\color[rgb]{0,0,0}\makebox(0,0)[lb]{\smash{\red\large$3$}}}%
    \put(0.81770306,0.55874671){\color[rgb]{0,0,0}\makebox(0,0)[lb]{\smash{\blue\Large$1$}}}%
    \put(0.81662384,0.24631531){\color[rgb]{0,0,0}\makebox(0,0)[lb]{\smash{\red\large$2$}}}%
    \put(0.81689797,0.02757381){\color[rgb]{0,0,0}\makebox(0,0)[lb]{\smash{\red\large$3$}}}%
    \put(0.01361111,0.14819441){\color[rgb]{0,0,0}\makebox(0,0)[lb]{\smash{\blue\Large$\bold{0}$}}}%
    \put(0.88381942,0.14819441){\color[rgb]{0,0,0}\makebox(0,0)[lb]{\smash{\blue\Large$\bold{0}$}}}%
    \put(0.32687018,-0.00863576){\color[rgb]{0,0,0}\makebox(0,0)[lb]{\smash{\red BSC($\frac{1}{2}-\gamma$)}}}%
    \put(0.32896425,0.46527616){\color[rgb]{0,0,0}\makebox(0,0)[lb]{\smash{\blue BSC($\frac{1}{2}-\eta$)}}}%
  \end{picture}%
\endgroup

%% file: example1_2.pdf_tex

\begingroup
  \makeatletter
  \providecommand\color[2][]{%
    \errmessage{(Inkscape) Color is used for the text in Inkscape, but the package 'color.sty' is not loaded}
    \renewcommand\color[2][]{}%
  }
  \providecommand\transparent[1]{%
    \errmessage{(Inkscape) Transparency is used (non-zero) for the text in Inkscape, but the package 'transparent.sty' is not loaded}
    \renewcommand\transparent[1]{}%
  }
  \providecommand\rotatebox[2]{#2}
  \ifx\svgwidth\undefined
    \setlength{\unitlength}{181.41732178pt}
  \else
    \setlength{\unitlength}{\svgwidth}
  \fi
  \global\let\svgwidth\undefined
  \makeatother
  \begin{picture}(1,0.87499996)%
    \put(0,0){\includegraphics[width=\unitlength]{example1_2.pdf}}%
    \put(0.09247512,0.7590963){\color[rgb]{0,0,0}\makebox(0,0)[lb]{\smash{\Large$X$}}}%
    \put(0.77539009,0.75947446){\color[rgb]{0,0,0}\makebox(0,0)[lb]{\smash{\Large$Y$}}}%
    \put(0.03125001,0.67211464){\color[rgb]{0,0,0}\makebox(0,0)[lb]{\smash{\large$1$}}}%
    \put(0.03083866,0.36103651){\color[rgb]{0,0,0}\makebox(0,0)[lb]{\smash{\large$2$}}}%
    \put(0.0292796,0.04697784){\color[rgb]{0,0,0}\makebox(0,0)[lb]{\smash{\large$3$}}}%
    \put(0.84416139,0.67339955){\color[rgb]{0,0,0}\makebox(0,0)[lb]{\smash{\large$1$}}}%
    \put(0.84308218,0.36096815){\color[rgb]{0,0,0}\makebox(0,0)[lb]{\smash{\large$2$}}}%
    \put(0.8433563,0.04521276){\color[rgb]{0,0,0}\makebox(0,0)[lb]{\smash{\large$3$}}}%
    \put(0.62674707,0.71561381){\color[rgb]{0,0,0}\makebox(0,0)[lb]{\smash{\small$\frac{1}{2}+\eta$}}}%
    \put(0.12940303,0.16993868){\color[rgb]{0,0,0}\makebox(0,0)[lb]{\smash{\small$\frac{1}{2}-\eta$}}}%
    \put(0.52228925,0.62830567){\color[rgb]{0,0,0}\makebox(0,0)[lb]{\smash{\small$\frac{1}{2}-\eta$}}}%
    \put(0.68600951,0.44238248){\color[rgb]{0,0,0}\makebox(0,0)[lb]{\smash{\small$\frac{1}{4}-\frac{1}{2}\eta$}}}%
    \put(0.09198528,0.56282518){\color[rgb]{0,0,0}\makebox(0,0)[lb]{\smash{\small$\frac{1}{4}-\frac{1}{2}\eta$}}}%
    \put(0.74523655,0.19956112){\color[rgb]{0,0,0}\makebox(0,0)[lb]{\smash{\small$(\frac{1}{2}+\eta)(\frac{1}{2}-\gamma)$}}}%
    \put(0.69534619,0.29310515){\color[rgb]{0,0,0}\makebox(0,0)[lb]{\smash{\small$(\frac{1}{2}+\eta)(\frac{1}{2}-\gamma)$}}}%
    \put(0.4334221,0.01714977){\color[rgb]{0,0,0}\makebox(0,0)[lb]{\smash{\small$(\frac{1}{2}+\eta)(\frac{1}{2}+\gamma)$}}}%
    \put(0.70625974,0.51917097){\color[rgb]{0,0,0}\makebox(0,0)[lb]{\smash{\small$(\frac{1}{2}+\eta)(\frac{1}{2}+\gamma)$}}}%
  \end{picture}%
\endgroup

%% file: example1_2_1.pdf_tex

\begingroup
  \makeatletter
  \providecommand\color[2][]{%
    \errmessage{(Inkscape) Color is used for the text in Inkscape, but the package 'color.sty' is not loaded}
    \renewcommand\color[2][]{}%
  }
  \providecommand\transparent[1]{%
    \errmessage{(Inkscape) Transparency is used (non-zero) for the text in Inkscape, but the package 'transparent.sty' is not loaded}
    \renewcommand\transparent[1]{}%
  }
  \providecommand\rotatebox[2]{#2}
  \ifx\svgwidth\undefined
    \setlength{\unitlength}{209.76376953pt}
  \else
    \setlength{\unitlength}{\svgwidth}
  \fi
  \global\let\svgwidth\undefined
  \makeatother
  \begin{picture}(1,0.72972974)%
    \put(0,0){\includegraphics[width=\unitlength]{example1_2_1.pdf}}%
    \put(0.72912151,0.03006988){\color[rgb]{0,0,0}\makebox(0,0)[lb]{\smash{$P_1 = [0 \ 1 \ 0]^T$}}}%
    \put(0.41575951,0.67955165){\color[rgb]{0,0,0}\makebox(0,0)[lb]{\smash{$P_0 = [1 \ 0 \ 0]^T$}}}%
    \put(0.00273355,0.02352606){\color[rgb]{0,0,0}\makebox(0,0)[lb]{\smash{$P_2 = [0 \ 0 \ 1]^T$}}}%
    \put(0.40540543,0.16216206){\color[rgb]{0,0,0}\makebox(0,0)[lb]{\smash{\small$P_{X|U=1}$}}}%
    \put(0.43243244,0.54054064){\color[rgb]{0,0,0}\makebox(0,0)[lb]{\smash{\small$P_{X|U=0}$}}}%
    \put(0.4724815,0.35310207){\color[rgb]{0,0,0}\makebox(0,0)[lb]{\smash{\small$P_X = [\frac{1}{2} \ \frac{1}{4} \ \frac{1}{4}]^T$}}}%
  \end{picture}%
\endgroup

%% file: layer_coding.pdf_tex

\begingroup
  \makeatletter
  \providecommand\color[2][]{%
    \errmessage{(Inkscape) Color is used for the text in Inkscape, but the package 'color.sty' is not loaded}
    \renewcommand\color[2][]{}%
  }
  \providecommand\transparent[1]{%
    \errmessage{(Inkscape) Transparency is used (non-zero) for the text in Inkscape, but the package 'transparent.sty' is not loaded}
    \renewcommand\transparent[1]{}%
  }
  \providecommand\rotatebox[2]{#2}
  \ifx\svgwidth\undefined
    \setlength{\unitlength}{362.83464355pt}
  \else
    \setlength{\unitlength}{\svgwidth}
  \fi
  \global\let\svgwidth\undefined
  \makeatother
  \begin{picture}(1,0.46874999)%
    \put(0,0){\includegraphics[width=\unitlength]{layer_coding.pdf}}%
    \put(0.14062499,0.40624994){\color[rgb]{0,0,0}\makebox(0,0)[lb]{\smash{$\overbrace{\thinspace\thinspace\thinspace\qquad\quad\quad\quad\quad\thinspace\thinspace\thinspace}$}}}%
    \put(0.37499997,0.09375004){\color[rgb]{0,0,0}\makebox(0,0)[lb]{\smash{\Large$\cdots$}}}%
    \put(0.46875,0.40624994){\color[rgb]{0,0,0}\makebox(0,0)[lb]{\smash{$\overbrace{\thinspace\thinspace\thinspace\qquad\quad\quad\quad\quad\thinspace\thinspace\thinspace}$}}}%
    \put(0.79687502,0.40624994){\color[rgb]{0,0,0}\makebox(0,0)[lb]{\smash{$\overbrace{\qquad\quad\quad\quad\quad\thinspace\thinspace\thinspace\thinspace\thinspace\thinspace}$}}}%
    \put(0.21874998,0.43749992){\color[rgb]{0,0,0}\makebox(0,0)[lb]{\smash{$u_1(1)$}}}%
    \put(0.15624998,0.34374998){\color[rgb]{0,0,0}\makebox(0,0)[lb]{\smash{$\sim P^*_{X|U_1=u_1(1)}$}}}%
    \put(0.80440337,0.34374998){\color[rgb]{0,0,0}\makebox(0,0)[lb]{\smash{$\sim P^*_{X|U_1=u_1(k_1)}$}}}%
    \put(0.48437502,0.34374998){\color[rgb]{0,0,0}\makebox(0,0)[lb]{\smash{$\sim P^*_{X|U_1=u_1(i)}$}}}%
    \put(0.375,0.35937497){\color[rgb]{0,0,0}\makebox(0,0)[lb]{\smash{\Large$\cdots$}}}%
    \put(-0.00000001,0.35937497){\color[rgb]{0,0,0}\makebox(0,0)[lb]{\smash{Layer 1:}}}%
    \put(0.20312498,0.37499996){\color[rgb]{0,0,0}\makebox(0,0)[lb]{\smash{$\ux(1)$}}}%
    \put(0.53125002,0.37499996){\color[rgb]{0,0,0}\makebox(0,0)[lb]{\smash{$\ux(i)$}}}%
    \put(0.85937498,0.37499996){\color[rgb]{0,0,0}\makebox(0,0)[lb]{\smash{$\ux(k_1)$}}}%
    \put(0.70312497,0.35937503){\color[rgb]{0,0,0}\makebox(0,0)[lb]{\smash{\Large$\cdots$}}}%
    \put(0.70312503,0.09375004){\color[rgb]{0,0,0}\makebox(0,0)[lb]{\smash{\Large$\cdots$}}}%
    \put(0.54687501,0.43749992){\color[rgb]{0,0,0}\makebox(0,0)[lb]{\smash{$u_1(i)$}}}%
    \put(0.85937498,0.43749992){\color[rgb]{0,0,0}\makebox(0,0)[lb]{\smash{$u_1(k_1)$}}}%
    \put(-0.00000003,0.09375004){\color[rgb]{0,0,0}\makebox(0,0)[lb]{\smash{Layer 2:}}}%
    \put(0.15624996,0.09375004){\color[rgb]{0,0,0}\makebox(0,0)[lb]{\smash{$P^*_{X|u_1(i),u_2(1)}$}}}%
    \put(0.48437495,0.09375004){\color[rgb]{0,0,0}\makebox(0,0)[lb]{\smash{$P^*_{X|u_1(i),u_2(j)}$}}}%
    \put(0.81250001,0.09375004){\color[rgb]{0,0,0}\makebox(0,0)[lb]{\smash{$P^*_{X|u_1(i),u_2(k_2)}$}}}%
    \put(0.14062498,0.06249992){\color[rgb]{0,0,0}\makebox(0,0)[lb]{\smash{$\underbrace{\thinspace\thinspace\thinspace\qquad\quad\quad\quad\quad\thinspace\thinspace\thinspace}$}}}%
    \put(0.46875,0.06249992){\color[rgb]{0,0,0}\makebox(0,0)[lb]{\smash{$\underbrace{\thinspace\thinspace\thinspace\qquad\quad\quad\quad\quad\thinspace\thinspace\thinspace}$}}}%
    \put(0.79687502,0.06249992){\color[rgb]{0,0,0}\makebox(0,0)[lb]{\smash{$\underbrace{\qquad\quad\quad\quad\quad\thinspace\thinspace\thinspace\thinspace\thinspace\thinspace}$}}}%
    \put(0.21874998,0.01562509){\color[rgb]{0,0,0}\makebox(0,0)[lb]{\smash{$u_2(1)$}}}%
    \put(0.54687501,0.01562509){\color[rgb]{0,0,0}\makebox(0,0)[lb]{\smash{$u_2(j)$}}}%
    \put(0.85937498,0.01562509){\color[rgb]{0,0,0}\makebox(0,0)[lb]{\smash{$u_2(k_2)$}}}%
    \put(0.50000004,0.2812499){\color[rgb]{0,0,0}\makebox(0,0)[lb]{\smash{$n_1$ symbols}}}%
    \put(0.46875003,0.328125){\color[rgb]{0,0,0}\makebox(0,0)[lb]{\smash{$\underbrace{\qquad\quad\quad\quad\quad\thinspace\thinspace\thinspace\thinspace\thinspace\thinspace}$}}}%
    \put(0.17187498,0.17187511){\color[rgb]{0,0,0}\makebox(0,0)[lb]{\smash{$n_2$ symbols}}}%
    \put(0.50000001,0.17187511){\color[rgb]{0,0,0}\makebox(0,0)[lb]{\smash{$n_2$ symbols}}}%
    \put(0.828125,0.17187511){\color[rgb]{0,0,0}\makebox(0,0)[lb]{\smash{$n_2$ symbols}}}%
    \put(0.14062498,0.14062514){\color[rgb]{0,0,0}\makebox(0,0)[lb]{\smash{$\overbrace{\thinspace\thinspace\thinspace\qquad\quad\quad\quad\quad\thinspace\thinspace\thinspace}$}}}%
    \put(0.46875,0.14062514){\color[rgb]{0,0,0}\makebox(0,0)[lb]{\smash{$\overbrace{\thinspace\thinspace\thinspace\qquad\quad\quad\quad\quad\thinspace\thinspace\thinspace}$}}}%
    \put(0.79687502,0.14062514){\color[rgb]{0,0,0}\makebox(0,0)[lb]{\smash{$\overbrace{\qquad\quad\quad\quad\quad\thinspace\thinspace\thinspace\thinspace\thinspace\thinspace}$}}}%
  \end{picture}%
\endgroup

%% file: example1_2_2.pdf_tex

\begingroup
  \makeatletter
  \providecommand\color[2][]{%
    \errmessage{(Inkscape) Color is used for the text in Inkscape, but the package 'color.sty' is not loaded}
    \renewcommand\color[2][]{}%
  }
  \providecommand\transparent[1]{%
    \errmessage{(Inkscape) Transparency is used (non-zero) for the text in Inkscape, but the package 'transparent.sty' is not loaded}
    \renewcommand\transparent[1]{}%
  }
  \providecommand\rotatebox[2]{#2}
  \ifx\svgwidth\undefined
    \setlength{\unitlength}{209.76376953pt}
  \else
    \setlength{\unitlength}{\svgwidth}
  \fi
  \global\let\svgwidth\undefined
  \makeatother
  \begin{picture}(1,0.72972974)%
    \put(0,0){\includegraphics[width=\unitlength]{example1_2_2.pdf}}%
    \put(0.59607396,0.02306749){\color[rgb]{0,0,0}\makebox(0,0)[lb]{\smash{$P_{X|U_1=1,U_2=0}= [0 \ 1 \ 0]^T$}}}%
    \put(0.36324075,0.68305284){\color[rgb]{0,0,0}\makebox(0,0)[lb]{\smash{$P_{X|U_1=0}=[1 \ 0 \ 0]^T$}}}%
    \put(-0.10930648,0.01827415){\color[rgb]{0,0,0}\makebox(0,0)[lb]{\smash{$P_{X|U_1=1,U_2=1} = [0 \ 0 \ 1]^T$}}}%
    \put(0.47564578,0.13206473){\color[rgb]{0,0,0}\makebox(0,0)[lb]{\smash{$P_{X|U_1=1}=[0 \ \frac{1}{2} \ \frac{1}{2}]^T$}}}%
    \put(0.4724815,0.35310207){\color[rgb]{0,0,0}\makebox(0,0)[lb]{\smash{\small$P_X = [\frac{1}{2} \ \frac{1}{4} \ \frac{1}{4}]^T$}}}%
  \end{picture}%
\endgroup

%% file: windmill_0.pdf_tex

\begingroup
  \makeatletter
  \providecommand\color[2][]{%
    \errmessage{(Inkscape) Color is used for the text in Inkscape, but the package 'color.sty' is not loaded}
    \renewcommand\color[2][]{}%
  }
  \providecommand\transparent[1]{%
    \errmessage{(Inkscape) Transparency is used (non-zero) for the text in Inkscape, but the package 'transparent.sty' is not loaded}
    \renewcommand\transparent[1]{}%
  }
  \providecommand\rotatebox[2]{#2}
  \ifx\svgwidth\undefined
    \setlength{\unitlength}{362.83464355pt}
  \else
    \setlength{\unitlength}{\svgwidth}
  \fi
  \global\let\svgwidth\undefined
  \makeatother
  \begin{picture}(1,0.39062498)%
    \put(0,0){\includegraphics[width=\unitlength]{windmill_0.pdf}}%
    \put(0,0.35937494){\color[rgb]{0,0,0}\makebox(0,0)[lb]{\smash{$X$}}}%
    \put(0.23437499,0.35937494){\color[rgb]{0,0,0}\makebox(0,0)[lb]{\smash{$Y_1$}}}%
    \put(0.35937501,0.35937494){\color[rgb]{0,0,0}\makebox(0,0)[lb]{\smash{$X$}}}%
    \put(0.59374997,0.35937494){\color[rgb]{0,0,0}\makebox(0,0)[lb]{\smash{$Y_2$}}}%
    \put(0.71875002,0.35937494){\color[rgb]{0,0,0}\makebox(0,0)[lb]{\smash{$X$}}}%
    \put(0.95312498,0.35937494){\color[rgb]{0,0,0}\makebox(0,0)[lb]{\smash{$Y_3$}}}%
    \put(0.17288708,0.29485072){\color[rgb]{0,0,0}\makebox(0,0)[lb]{\smash{$\frac{1}{2}$}}}%
    \put(0.0308062,0.25854049){\color[rgb]{0,0,0}\makebox(0,0)[lb]{\smash{$\frac{1}{2}$}}}%
    \put(0.74999999,0.07812487){\color[rgb]{0,0,0}\makebox(0,0)[lb]{\smash{$\frac{1}{2}$}}}%
    \put(0.39872164,0.12196372){\color[rgb]{0,0,0}\makebox(0,0)[lb]{\smash{$\frac{1}{2}$}}}%
    \put(0.3997337,0.20616131){\color[rgb]{0,0,0}\makebox(0,0)[lb]{\smash{$\frac{1}{2}$}}}%
    \put(0.921875,0.04181464){\color[rgb]{0,0,0}\makebox(0,0)[lb]{\smash{$\frac{1}{2}$}}}%
    \put(0.22482248,0.19907666){\color[rgb]{0,0,0}\makebox(0,0)[lb]{\smash{\small$\delta$}}}%
    \put(0.16479043,0.08634627){\color[rgb]{0,0,0}\makebox(0,0)[lb]{\smash{\small$\delta$}}}%
    \put(0.54586294,0.27720174){\color[rgb]{0,0,0}\makebox(0,0)[lb]{\smash{\small$\delta$}}}%
    \put(0.53485452,0.04137082){\color[rgb]{0,0,0}\makebox(0,0)[lb]{\smash{\small$\delta$}}}%
    \put(0.93952412,0.12848006){\color[rgb]{0,0,0}\makebox(0,0)[lb]{\smash{\small$\delta$}}}%
    \put(0.88309668,0.23108981){\color[rgb]{0,0,0}\makebox(0,0)[lb]{\smash{\small$\delta$}}}%
    \put(0.58221133,0.19400039){\color[rgb]{0,0,0}\makebox(0,0)[lb]{\smash{\small$1-\delta$}}}%
    \put(0.90268346,0.2765469){\color[rgb]{0,0,0}\makebox(0,0)[lb]{\smash{\small$1-\delta$}}}%
    \put(0.83141016,0.08425163){\color[rgb]{0,0,0}\makebox(0,0)[lb]{\smash{\small$1-\delta$}}}%
    \put(0.58567519,0.12271074){\color[rgb]{0,0,0}\makebox(0,0)[lb]{\smash{\small$1-\delta$}}}%
    \put(0.11321234,0.23455361){\color[rgb]{0,0,0}\makebox(0,0)[lb]{\smash{\small$1-\delta$}}}%
    \put(0.16337256,0.04003921){\color[rgb]{0,0,0}\makebox(0,0)[lb]{\smash{\small$1-\delta$}}}%
  \end{picture}%
\endgroup

%% file: windmill_1.pdf_tex

\begingroup
  \makeatletter
  \providecommand\color[2][]{%
    \errmessage{(Inkscape) Color is used for the text in Inkscape, but the package 'color.sty' is not loaded}
    \renewcommand\color[2][]{}%
  }
  \providecommand\transparent[1]{%
    \errmessage{(Inkscape) Transparency is used (non-zero) for the text in Inkscape, but the package 'transparent.sty' is not loaded}
    \renewcommand\transparent[1]{}%
  }
  \providecommand\rotatebox[2]{#2}
  \ifx\svgwidth\undefined
    \setlength{\unitlength}{119.05511475pt}
  \else
    \setlength{\unitlength}{\svgwidth}
  \fi
  \global\let\svgwidth\undefined
  \makeatother
  \begin{picture}(1,0.85714286)%
    \put(0,0){\includegraphics[width=\unitlength]{windmill_1.pdf}}%
    \put(0.55398042,0.52380947){\color[rgb]{0,0,0}\makebox(0,0)[lb]{\smash{\footnotesize$\sigma_1 \cdot \K$}}}%
    \put(0.73231456,0.38745135){\color[rgb]{0,0,0}\makebox(0,0)[lb]{\smash{\scriptsize$\| B_1 \K \|$}}}%
    \put(0.03528131,0.39020538){\color[rgb]{0,0,0}\makebox(0,0)[lb]{\smash{\scriptsize$\| B_2\K \|$}}}%
    \put(0.19875641,0.44918952){\color[rgb]{0,0,0}\makebox(0,0)[lb]{\smash{\scriptsize$\| B_3 \K \|$}}}%
    \put(0.42639041,0.47655424){\color[rgb]{0,0,0}\makebox(0,0)[lb]{\smash{$\frac{2\pi}{3}$}}}%
    \put(0.33987642,0.36059642){\color[rgb]{0,0,0}\makebox(0,0)[lb]{\smash{$\frac{2\pi}{3}$}}}%
    \put(0.93954485,0.41605426){\color[rgb]{0,0,0}\makebox(0,0)[lb]{\smash{$\uv_1$}}}%
    \put(0.40351795,0.80314684){\color[rgb]{0,0,0}\makebox(0,0)[lb]{\smash{$\uv_2$}}}%
  \end{picture}%
\endgroup

%% file: windmill_2.pdf_tex

\begingroup
  \makeatletter
  \providecommand\color[2][]{%
    \errmessage{(Inkscape) Color is used for the text in Inkscape, but the package 'color.sty' is not loaded}
    \renewcommand\color[2][]{}%
  }
  \providecommand\transparent[1]{%
    \errmessage{(Inkscape) Transparency is used (non-zero) for the text in Inkscape, but the package 'transparent.sty' is not loaded}
    \renewcommand\transparent[1]{}%
  }
  \providecommand\rotatebox[2]{#2}
  \ifx\svgwidth\undefined
    \setlength{\unitlength}{110.55118408pt}
  \else
    \setlength{\unitlength}{\svgwidth}
  \fi
  \global\let\svgwidth\undefined
  \makeatother
  \begin{picture}(1,0.8461538)%
    \put(0,0){\includegraphics[width=\unitlength]{windmill_2.pdf}}%
    \put(0.76923076,0.61538441){\color[rgb]{0,0,0}\makebox(0,0)[lb]{\smash{$\phi_{\theta}$}}}%
    \put(0.22900762,0.72695239){\color[rgb]{0,0,0}\makebox(0,0)[lb]{\smash{$\phi_{\theta+\frac{2\pi}{3}}$}}}%
    \put(0.39107456,0.07633607){\color[rgb]{0,0,0}\makebox(0,0)[lb]{\smash{$\phi_{\theta+\frac{4\pi}{3}}$}}}%
  \end{picture}%
\endgroup

%% file: example3.pdf_tex

\begingroup
  \makeatletter
  \providecommand\color[2][]{%
    \errmessage{(Inkscape) Color is used for the text in Inkscape, but the package 'color.sty' is not loaded}
    \renewcommand\color[2][]{}%
  }
  \providecommand\transparent[1]{%
    \errmessage{(Inkscape) Transparency is used (non-zero) for the text in Inkscape, but the package 'transparent.sty' is not loaded}
    \renewcommand\transparent[1]{}%
  }
  \providecommand\rotatebox[2]{#2}
  \ifx\svgwidth\undefined
    \setlength{\unitlength}{164.40944824pt}
  \else
    \setlength{\unitlength}{\svgwidth}
  \fi
  \global\let\svgwidth\undefined
  \makeatother
  \begin{picture}(1,0.44827586)%
    \put(0,0){\includegraphics[width=\unitlength]{example3.pdf}}%
    \put(0,0.34482753){\color[rgb]{0,0,0}\makebox(0,0)[lb]{\smash{\Large$X_1$}}}%
    \put(0,0.03448293){\color[rgb]{0,0,0}\makebox(0,0)[lb]{\smash{\Large$X_2$}}}%
    \put(0.83074493,0.17386169){\color[rgb]{0,0,0}\makebox(0,0)[lb]{\smash{\Large$Y$}}}%
  \end{picture}%
\endgroup